\documentclass[10pt, a4paper, twoside, titlepages]{article}
\usepackage[utf8]{inputenc}
\usepackage{amsmath,color,oldgerm,mathrsfs}
\usepackage{amssymb}

\usepackage[utf8]{inputenc}
\usepackage[english]{babel}
 
\usepackage{amsthm}

\usepackage{systeme}

\usepackage{geometry}
\usepackage{multirow}
\usepackage{array}
\usepackage{dsfont}

\usepackage{pgffor}
\usepackage{graphicx}
\usepackage{subcaption}

\usepackage[section]{placeins}
\usepackage{float}
\usepackage{python}

\usepackage{verbatim}
\usepackage{enumitem}
\usepackage[multiple]{footmisc}

\newcolumntype{M}{>{\centering\arraybackslash}p{1.3cm}}

\newtheorem{theorem}{Théorem}[section]

\newtheorem{remark}[theorem]{Remark}
\newtheorem{corollary}[theorem]{Corollary}

\newtheorem{proposition}[theorem]{Proposition}
\newtheorem{assumption}{Assumption}

\makeatletter

\makeatletter
\newcommand\nocaption{%
    \renewcommand\p@subfigure{}
    \renewcommand\thesubfigure{\thefigure\alph{subfigure}}
}
\makeatother

\def\Nmc{\Nmc}
\def\Ns{n_{s}}
\def\Nw{n_{w}}
\def\Eb{{\mathbb E}}
\def\Pb{{\mathbb P}}
\def\Nb{{\mathbb N}}
\def\Rb{{\mathbb R}}
\def\Sb{{\mathbb S}}
\def\Lb{{\mathbb L}}
\def\Qc{{\mathcal Q}}
\def\Nc{{\mathcal N}}
\def\Fc{{\mathcal F}}
\def\Ac{{\mathcal A}}
\def\Uc{{\mathcal U}}
\def\Mc{{\mathcal M}}

\def\ES{{\rm ES}}
\def\VaR{{\rm VaR}}
\def\ESh{\widehat{\rm ES}}

\def\mk{{\mathfrak m}}

\def\Ik{{\mathfrak I}}

\def\kk{{\mathfrak k}}

\def\1{{\mathbf 1}}

\def\Nmc{N}
\def\sr{{\rm s}}
\def\Frm{{\rm F}}

\def\ps{\mathtt p}
\def\is{{\mathtt i}}
\def\ks{{\mathtt k}}
\def\Ss{{\mathtt S}}
\def\ms{{\mathtt m}}
\def\[#1{[\![#1]\!]}
\def\essinf#1{{\rm ess}\!\!\!\!\inf_{#1}}

\def\E{{\mathbb E}}
\def\vs#1{\vspace{#1mm}}
\def\x{{\rm x}}
\def\X{{\rm X}}
\title{Computation of Expected Shortfall by fast detection of worst scenarios}
\author{Bruno Bouchard\footnote{CEREMADE,  CNRS, Universit\'e Paris Dauphine,  PSL University.}, Adil Reghai\footnote{Natixis.}, Benjamin Virrion\footnote{Natixis and CEREMADE, CNRS, Universit\'e Paris Dauphine, PSL University.}$^{\; \: ,}$\footnote{The authors would like to thank Nicolas Baradel for helping with the code, Rida Mahi and Mathieu Bernardo from the Natixis Quantitative Research Teams for providing the first results and ideas on the Fast Detection Algorithm, and finally William Leduc for providing all the necessary data to obtain the different book parameters.}}
\date{\today}
\begin{document}

\maketitle

\begin{abstract}
We consider a multi-step algorithm for the computation of the historical expected shortfall such as defined by the Basel {\sl Minimum Capital Requirements for Market Risk}. 
At each step of the algorithm, we use Monte Carlo simulations to reduce the number of historical scenarios that potentially belong to the set of worst scenarios. The number of simulations increases as the number of candidate scenarios is reduced and the distance between them diminishes. For the most naive scheme, we show that the ${\Lb}^{p}$-error of the estimator of the Expected Shortfall is bounded by a linear combination of the probabilities of inversion of favorable and unfavorable scenarios at each step, and of the last step Monte Carlo error associated to each scenario. By using concentration inequalities, we then show that, for sub-gamma pricing errors, the probabilities of inversion converge at an exponential rate in the number of simulated paths. We then propose an adaptative version in which the algorithm improves step by step its knowledge on the unknown parameters of interest: mean and variance of the Monte Carlo estimators of the different scenarios. Both schemes can be optimized by using dynamic programming algorithms that can be solved off-line.
To our knowledge, these are the first non-asymptotic bounds for such estimators. Our hypotheses are weak enough to allow for the use of estimators for the different scenarios and steps based on the same random variables, which, in practice, reduces considerably the computational effort. First numerical tests are performed. 
\end{abstract}


\section{Introduction}

The Basel {\sl Minimum Capital Requirements for Market Risk} \cite{basel_minimum_capital_requirements_for_market_risk} has brought two main changes in the way that investment banks need to compute their capital requirements. 
Expected Shortfall ($\ES$) replaces Value at Risk (VaR) as the main risk indicator for the computation of capital requirements. The advantages of ES over VaR have been brought forward in Artzner et al.~\cite{artzner1999coherent}, and Expected Shortfall is now considered by most researchers and practitioners as superior to VaR as a risk measure, because it respects the sub-additivity axiom, see \cite{acerbi2002coherence,artzner1999coherent,rockafellar2002conditional}. 
The second main change is that the number of required daily computations of $\ES$ has been multiplied by a factor of up to 90. Where banks used to need to compute one VaR per day, they now need to compute up to three $\ES$ per liquidity horizon and risk class, as well as three $\ES$ per liquidity horizon for all risk classes combined.
The above has triggered several works on the fast computation of $\ES$. 
\vs2

Mathematically, if {$V$} is a random variable modeling the level of loss\footnote{All over this paper, we measure the performances in terms of losses. A positive number is a loss, a negative number is a gain.} of a portfolio that will be known at a future time, and $0 < \alpha < 1$, the expected shortfall of level $\alpha\in (0,1)$ is defined by
\begin{equation}
\ES_\alpha := \frac{1}{\alpha} \int_0^\alpha \VaR_\gamma( {V}) d\gamma,
\end{equation} 
where $\VaR_{\gamma}$ is the Value at Risk at level $\gamma$, i.e.
\begin{equation}
\VaR_\gamma \left( {V}\right) :=  {\max \left\{ x \in \mathbb{R} : \Pb[V\ge x] > \gamma \right\}}. 
\end{equation}

Nearly all of the literature concentrates on studying the $\ES$ by using parametric, non-parametric or semi-parametric approaches to approximate the distribution of ${V}$ based on historical data. See in particular \cite{broda2018approximating,elliott2009var,francq2014multi,hoogerheide2010bayesian,krause2014fast,nadarajah2014estimation,ortiz2014efficient,peracchi2008estimating,kamdem2005value,simonato2011performance,yu2010kernel,yueh2010analytical}. Another approach consists in using the fact that $V$ is the risk neutral value of a book, and therefore of the form $\mathbb{E}[P|S]$ in which $S$ is a random variable associated to market parameters and $P$ represents the future (discounted) payoffs of the book. This suggests using a nested Monte Carlo approach : simulate a set of values in the distribution of $S$ (outer scenarios), and, for each simulation of $S$, compute a Monte Carlo estimator of $\E[P|S]$ by using simulations in the conditional distribution (inner scenarios). This is for instance the approach of \cite{broadie2011efficient,gordy2010nested}.

But, as defined in the regulatory document of Basel \cite{basel_minimum_capital_requirements_for_market_risk}, the expected shortfall is based on $\Ns = 253$ scenarios of market parameters $\sr=(\sr^{i})_{i\le \Ns}$ that are generated in a deterministic way. Therefore, $S$ is just uniformly distributed in the sequence $\sr$ and there is no need for simulating outer scenarios. Since $V$ is defined by a pricing formula $\mathbb{E}[P|S]$ that is fully described by the value of $S$, there is also no room for approximating the law of $V$ based on historical data, if we are only interested by the requirements of \cite{basel_minimum_capital_requirements_for_market_risk}. 
The only issue is to compute in an efficient way the loss impacts $(\mu^{i})_{i\le \Ns}$ of the book,
\begin{equation*}
    \mu^{i} :=   \left(\mathbb{E}[P | S=\sr^{i}]-\mathbb{E}[P | S=\sr^{0}]\right), \;i=1 \ldots,\Ns,
\end{equation*}
in which $\sr^{0}$ is the current value of the market parameters, and then compute the average over the $\Nw = 6$ worst impacts, say
\begin{equation}
\ES = \frac{1}{\Nw} \sum_{i=1}^{\Nw} \mu^{i},
\end{equation}
if, for ease of notations, we assume that 
\begin{equation}\label{eq: ordre mu}
\mu^1\ge \mu^2 \ge \cdots\ge  \mu^{\Ns - 1}\ge \mu^{\Ns}.
\end{equation}

Methods that are in line with the above have also been studied, in particular in \cite{liu2010stochastic,risk2018sequential} in which the authors define a distance on the space of scenarios induced by the distance between their risk factors. Starting with the original outer-level scenarios (called ``prediction points"), they determine ``design points" that are included in their convex hull. Inner-level paths are simulated in order to evaluate the portfolio value at the design points. These values are then used to establish a metamodel of the portfolio price with respect to the risk factors, and this metamodel is then used to select among the prediction points those that are most likely to be part of the worst scenarios set. They are then added to the design points, and evaluated by using inner-level simulations, after which the metamodel is updated. 

These methods are very smart but neglect one important point for practitioners: the cost of launching a pricer is high, as it typically entails instanciating thousands of objects at initialization, as well as volatility surface calibrations and sometimes even graphical interfaces. Furthermore, these pricers usually do not have the flexibility to add dynamically, at each inner-level pricing, new paths to a given scenario. Therefore, we do not allow ourselves to adapt our strategies at such a level of granularity. 

Instead, we will consider strategies that only entail $L$-levels of sets of simulations, where $L$ is typically quite low, so as not to pay too many times the overhead of launching the pricer and/or calibrating the required volatility surfaces. 
We also do not use any concept of distance between scenarios induced by their risk factors. Although this enables \cite{liu2010stochastic} and \cite{risk2018sequential} to obtain better empirical convergence rates, we see at least one problem with this approach: at the scale of a bank, the space of risk factors is both of a very high dimension (a few thousands) and with a very complex geometry (the payoffs of the portfolio's derivative products are usually non-convex, and path-dependent), so that it is very difficult to establish a model describing the proximity of scenarios in a robust way. \\

We thus study a relatively simple procedure that also has the advantage of allowing us to establish non-asymptotic bounds on the $\mathbb{L}^{p}$-error of our estimator, in the spirit of the simplest ranking by mean procedures, see e.g.~\cite{bahadur1950problem,bechhofer1954single,bechhofer1954tow,gupta1991sequential}. It consists in using a first set of simulated paths to provide a crude estimation of the impact factors $\mu^{i}$. These first estimators are ranked to select the $q_{1}<\Ns$ outer-level scenarios with the highest estimated impact values. Then, only the impact values of these $q_{1}$ pre-selected scenarios are estimated again by using the previous estimators together with a new set of simulated paths. Among these new estimators we select the scenarios with the $q_{2}<q_{1}$ highest estimated impact factors. And so on. After $L\ge 2$ steps, $L$ being small in practice, we just keep the mean of the six highest estimated impacts. 

The rationale behind this is that a first crude estimation should be sufficient to rule out a large part of the scenarios from the candidates of being in the $6$ worst ones, because the corresponding values should be far enough. While the number of candidates reduces, one can expect that the differences between the corresponding impacts diminish as well and that more Monte Carlo simulations are needed to differentiate them. Under an indifference zone hypothesis, similar to the one used in the above mentioned paper, and a sub-gamma distribution assumption, the convergence is exponential in the number of simulations used at the different steps and of order $1/2$ in the total number of simulations. See Proposition \ref{prop: borne Lp generale} and Corollary \ref{cor: borne si sub gamma} below. 

The optimal number of additional paths that should be used at each step to minimize the strong estimation error, given a maximal computational cost, can be determined by a simple dynamic programming algorithm, that can be solved off-line, see Section \ref{subsec: DPP determinist}. In theory, this requires the a priori knowledge of the means and covariances of our estimators, which are obviously not known in practice. However, one can easily define a version based on a robust specification of the error. One can also take advantage of the different simulation sets to improve our prior on the true hidden parameters. This leads to a natural adaptative algorithm, see Section \ref{sec : adaptative}, for which convergence is also proved, see Proposition \ref{prop: convergence cas strat random}. Estimating the optimal policy associated to this adaptative algorithm is costly but can be done off-line by using a neural network approximation combined with a backward dynamic programming algorithm. We explain how this can be done in Section \ref{subsec: neural network} (further details are in Appendix \ref{sec: precise implementation}). 

The rest of the paper is organized as follows. Section \ref{sec: cas determinist} is dedicated to the most naive deterministic algorithm. In particular, Section \ref{sec: 2 Level} gives a very easy to use two levels algorithm for the case where the impacts decrease linearly in the scenarios' rank order. The adaptative version of the algorithm is presented in Section \ref{sec : adaptative}. Finally, we perform first numerical tests in Section \ref{sec: tests numeriques}. 
 
\section{Algorithm with a deterministic effort allocation}\label{sec: cas determinist}

In this section, we describe the simplest version of the algorithm. It uses a pre-defined deterministic number of simulations. We establish a strong $\Lb^{p}$-error bound and discuss several ways of choosing the optimal strategy for minimizing this error.

\subsection{The algorithm}\label{subsec: deterministic algorithm}

From now on, we assume that $\Eb[P|S=\sr^{0}]$ is known perfectly and set it to $0$ for ease of notations. As explained above, the algorithm relies on the idea of selecting progressively the scenarios that will be used for the computation of the Expected Shortfall. Namely, let $P_{|\rm s}:=(P_{|\rm s^{i}})_{i\le \Ns}$ be a $\Ns$-dimensional random variable such that each $P_{|{\rm s}^{i}}$ has the law of $P$ given $S={\rm s}^{i}$. 
We first simulate independent copies $(P_{j}^{1}, \ldots, P_{j}^{\Ns})_{j\ge 1}$ of $P_{|\rm s}$ and compute the Monte Carlo estimators of $\Eb[P|S=\sr^{i}]$, $i\le \Ns$: 
$$
\hat \mu^{i}_{1}:=\frac1{\Nmc_{1}}\sum_{j=1}^{\Nmc_{1}} P_{j}^{i}\;\mbox{ for } i\le \Ns,
$$ 
for some $\Nmc_{1}\ge 1$.
Among these random variables, we then select the ones that are the most likely to coincide with the worst scenarios $\sr^{1},\ldots,\sr^{\Nw}$, for some $1\le \Nw<\Ns$. 
To do this, one considers the (random) permutation $\mk_{1}$ on $[\![1,\Ns]\!]$ such that the components of $\left(\hat \mu_{1}^{\mk_{1}(i)}\right)_{i\le \Ns}$ are in decreasing order:
\begin{equation*} 
\left\{
\begin{array} {l}
\hat \mu_{1}^{\mk_{1}(1)} \ge  \hat \mu_{1}^{\mk_{1}(2)} \ge \ldots \ge  \hat \mu_{1}^{\mk_{1}(\Ns)},  
\\
\mk_{1}(i)<\mk_{1}(i')\;\mbox{ if } \hat \mu_{1}^{\mk_{1}(i)} =  \hat \mu_{1}^{\mk_{1}(i')}\mbox{ for } 1\le i< i'\le \Ns,
\end{array}\right.
\end{equation*}
and only keep the indexes $(\mk_{1}(\ell))_{\ell \le q_{1}}$ of the corresponding $q_{1}\ge \Nw$ highest values, i.e.~the indexes belonging to 
$$
 \Ik_{1}:= \Ik_{0}\cap   \mk_{1}( [\![1,q_{1}]\!])\;\mbox{ in which }  \Ik_0:= [\![1,\Ns]\!].
$$ 
We then iterate the above procedure on the scenarios in $\Ik_{1}$ and so on. Namely, we fix $L\ge 1$ different thresholds $(q_{\ell})_{\ell=0,\ldots,L {-1}}$ such that 
\begin{align}\label{eq; cond sur q}
 {\Nw=: q_{L-1}}  {\le }\cdots {\le }q_{0}:=\Ns.
\end{align}
Assuming that $  \Ik_{\ell-1}$ is given, for some $1\le \ell-1\le L- {1}$, we compute the estimators\footnote{Note from the considerations below that only the elements $(\hat \mu^{i}_{\ell})_{i\in  \Ik_{\ell-1}}$ are needed in practice, the others are only defined here because they will be used in our proofs.}  
\begin{equation}\label{eq:mu_hat_def}
\hat \mu^{i}_{\ell}:=\frac1{\Nmc_{\ell}}\sum_{j=1}^{\Nmc_{\ell}} P_{j}^{i}\;\mbox{ for } i\le \Ns, 
\end{equation}
for some $\Nmc_{\ell}\ge \Nmc_{\ell-1}$.
If $\ell\le L-1$, we consider the (random) permutation $\mk_{\ell}:[\![1,q_{\ell-1}]\!]\mapsto \Ik_{\ell-1}$ such that the components of $\left(\hat \mu_{\ell}^{i}\right)_{i\in   \Ik_{\ell-1}  }$ are in decreasing order
\begin{equation}\label{eq: def m}
\left\{
\begin{array} {l}
\hat \mu_{\ell}^{\mk_{\ell}(1)} \ge \hat \mu_{\ell}^{\mk_{\ell}(2)} \ge \ldots \ge \hat \mu_{\ell}^{\mk_{\ell}(q_{\ell-1})},
\\
\mk_{\ell}(i)<\mk_{\ell}(i')\;\mbox{ if } \hat \mu_{\ell}^{\mk_{\ell}(i)} =  \hat \mu_{\ell}^{\mk_{\ell}(i')}\mbox{ for } 1\le i< i'\le \Ns,
\end{array}\right.
\end{equation}
and only keep the elements in 
$$
\Ik_{\ell}:=   \Ik_{\ell-1} \cap \mk_{\ell}([\![1,q_{\ell}]\!])
$$
for the next step. If $\ell=L$, we just compute the final estimator of the $\ES$ given by 
\begin{equation*}
\ESh := \frac{1}{\Nw} \sum_{i=1}^{\Nw} \hat \mu_{L}^{\mk_{ {L-1}}(i)}= \frac{1}{\Nw} \sum_{i\in   \Ik_{ {L-1}}} \hat \mu_{L}^{i}.
\end{equation*}
 
Note that only the $L-1$-first steps are used to select the worth scenarios, the step $L$ is a pure Monte Carlo step. Again, the general idea is to reduce little by little the number of candidate scenarios to be part of the worst ones. As the number of candidates diminishes, one increases the number of simulated paths so as to reduce the variance of our Monte Carlo estimators and be able to differentiate between potentially closer true values of the associated conditional expectations. \\
 
 \begin{remark} Note that, given $j$, we do not assume that the $P^{i}_{j}$, $i\le \Ns$, are independent. The simulations associated to different scenarios are in general not independent. Moreover, the $\hat \mu^{i}_{\ell}$, $\ell \le L$, use the same simulated paths, only the number of used simulations changes. Both permit to reduce the computational cost, by allowing the use of the same simulations of the underlying processes across scenarios and steps.
 \end{remark}
 

\subsection{General a-priori bound on the $\Lb^{p}$ error}
In this section, we first provide a general $\Lb^{p}$ estimate of the error. A more tractable formulation will be provided in Corollary \ref{cor: borne si sub gamma} under an additional sub-gamma distribution assumption. 
\vs2

From now on, we assume that $P_{|{\rm s}} \in \Lb^{p}$ for all $p\ge 1$, and we use the notations\footnote{The element $q_{L}$ and $\Nmc_{0}$ are defined for notational convenience, they never appear in our algorithm. To fix ideas, they can be set to $q_{L}=\Nw$ and $N_{0}=0$ all over this paper.} 
\begin{align}
&q:=( {q_{0}}, q_{1},\ldots,q_{ {L}}) \;\mbox{ , }\;\Nmc=(\Nmc_{0},\Nmc_{1},\ldots,\Nmc_{L})\nonumber\\
 &\delta q_{\ell}:=q_{\ell-1}-q_{\ell}\;\mbox{ and }\; \delta N_{\ell}:=N_{\ell}-N_{\ell-1} \mbox{, for } 1\le \ell\le L, \label{eq: def delta N}\\
  &\delta \hat \mu^{i}_{\ell}:=\frac{\sum_{j=N_{\ell-1}+1}^{N_{\ell}} P^{i}_{j}}{\delta \Nmc_{\ell}}=\frac{\Nmc_{\ell}\hat \mu^{i}_{\ell}-\Nmc_{\ell-1} \hat \mu^{i}_{\ell-1}}{\delta \Nmc_{\ell}}, \mbox{ for $1 \leq i \leq \Ns$,}\label{eq : def delta hat mu i}
\end{align} 
with the convention $0/0=0$. 

\begin{proposition}{\label{prop: borne Lp generale}} For all $p\ge 1$, 
\begin{align}
    \mathbb{E}\left[ \left|\ES - \ESh\right|^p \right]^{\frac{1}{p}} \leq & 
    \sum_{\ell=1}^{L-1} (\delta q_{\ell})^{\frac{1}{p}} \max_{(i,k)\in [\![1, {\Nw}]\!]\times [\![q_{\ell}+1,\Ns]\!]}  (\mu^{i}-\mu^{k})\Pb[\hat \mu_{\ell}^{k}>\hat \mu_{\ell}^{i}]^{\frac{1}{p}} \nonumber\\
    &+ \frac{1}{\Nw} \frac{\delta N_{L}}{N_{L}} \underset{1 \leq i_1 <\cdots < i_{ {\Nw}} \leq \Ns}{\max} 
    \left( \sum_{j=1}^{ {\Nw}}  \Eb\left[\left| \delta \hat \mu^{i_j}_{L} -  \mu^{i_j} \right|^p\right]^{\frac{1}{p}} \right)   \label{eq: borne Lp generale}\\
    &+ \frac{1}{\Nw} \frac{N_{L-1}}{N_{L}} \sum_{i=1}^{ {\Ns}}  \Eb\left[\left|  \hat \mu_{L-1}^{ {i}} - \mu^{ {i}} \right|^p\right]^{\frac{1}{p}} . \nonumber
\end{align}
\end{proposition}

Before providing the proof of this general estimate, let us make some comments. The last two terms in \eqref{eq: borne Lp generale} are natural as they are due to the Monte Carlo error made on the estimation of the various conditional expectations that can enter, after the $(L-1)$-levels selection procedure, in the estimation of $\ES$. Note that it corresponds to the estimation errors using the cumulated number of Monte Carlo simulations $\Nmc_{L-1}$ of step $L-1$ and the number $\Nmc_{L}-\Nmc_{L-1}$ of simulations used only for the last step. In practice, these numbers should be sufficiently large. The first term involves the quantities $ \max_{(i,k)\in [\![1, {\Nw}]\!]\times [\![q_{\ell}+1,\Ns]\!]}  (\mu^{i}-\mu^{k})\Pb[\hat \mu_{\ell}^{k}>\hat \mu_{\ell}^{i} ]^{\frac{1}{p}}$ with $\ell=1,\ldots,L-1$. Each term corresponds to the situation in which an element $i\in \[{1,\Nw}$ gets out the set of selected indexes $\Ik_{\ell}$ exactly at the $\ell$-th step. In the worst situation, it is replaced by an element of index $k$ larger than $q_{\ell}$ and this can happen only if $\hat \mu_{\ell}^{k}>\hat \mu_{\ell}^{i}$. The probability of this event is controlled by the number of Monte Carlo simulations $\Nmc_{\ell}$ used at the step $\ell$ but also by the {\sl distance} between the two scenarios. More specifically, for $\ell$ small, one expects that $\Pb[\hat \mu_{\ell}^{k}>\hat \mu_{\ell}^{i} ]$ is small because the law of $P_{|{\sr_{k}}}$ is concentrated far away from where the law of $P_{|{\sr_{i}}}$ is. This quantity potentially increases with $\ell$, as we reduce the number of selected indexes. This should be compensated by an increase in the number of used Monte Carlo simulations. Otherwise stated, we expect to balance the various terms of 
\eqref{eq: borne Lp generale} by considering a suitable increasing sequence $(\Nmc_{\ell})_{\ell\le L}$. 

Obviously, \eqref{eq: borne Lp generale} implies that the algorithm converges as $N_{\ell}\to \infty$ for all $\ell\le L$, see Proposition \ref{prop: convergence cas strat random} below for a proof in a more general framework. 

\begin{proof}[Proof of Proposition \ref{prop: borne Lp generale}] 
We split the error into a permutation and a Monte Carlo error: 
\begin{align}\label{eq: decompo ES - hat ES avec hat mu}
\mathbb{E}\left[ \left|\ES - \ESh\right|^p \right]^{\frac{1}{p}} 
\le&
\Eb\left[\left|\frac{1}{\Nw} \sum_{i\le \Nw} \mu^{i} -   \mu^{\mk_{L-1}(i)} \right|^p\right]^{\frac{1}{p}} 
+ \Eb\left[\left|\frac{1}{\Nw} \sum_{i\le \Nw}  \hat \mu_{L}^{\mk_{L-1}(i)}-   \mu^{\mk_{L-1}(i)} \right|^p\right]^{\frac{1}{p}}.
\end{align}
Let us first look at the second term which corresponds to a Monte Carlo error. We have
\begin{align*}
 \Eb\left[\left|\frac{1}{\Nw} \sum_{i\le \Nw} \hat \mu_{L}^{\mk_{L-1}(i)}- \mu^{\mk_{L-1}(i)} \right|^p\right]^{\frac{1}{p}}
\leq &\frac{N_{L-1}}{N_{L}} \frac{1}{\Nw} \Eb\left[\left| \sum_{i \le \Nw} \hat \mu_{L-1}^{\mk_{L-1}(i)}- \mu^{\mk_{L-1}(i)} \right|^p \right]^{\frac{1}{p}} \\
& + \frac{N_L - N_{L - 1}}{N_{L}} \frac{1}{\Nw} \Eb \left[ \left| \sum_{i \le \Nw} \frac{\sum_{j = N_{L-1} + 1}^{N_{L}}\hat P_j^{\mk_{L-1}(i)}}{N_{L} - N_{L-1}} - \mu^{\mk_{L-1}(i)} \right|^p \right]^{\frac{1}{p}}
\end{align*}
in which 
\begin{align*}
 \Eb\left[\left|\sum_{i\le \Nw}  \hat \mu_{L - 1}^{\mk_{L - 1}(i)}-   \mu^{\mk_{L - 1}(i)} \right|^p\right]^{\frac{1}{p}} 
 &\le \sum_{i\le \Ns}  \Eb\left[\left|  \hat \mu_{L - 1}^{i}-   \mu^{i} \right|^p\right]^{\frac{1}{p}},
\end{align*} 
and
\begin{equation*}
\begin{split}
\Eb \left[\left|\sum_{i \le \Nw} \frac{\sum_{j = N_{L-1} + 1}^{N_{L}}\hat P_j^{\mk_{L-1}(i)}}{N_{L} - N_{L-1}} - \mu^{\mk_{L-1}(i)} \right|^p  \right]^{\frac{1}{p}}
& = \Eb \left[ \Eb \left[ \left| \sum_{i \le \Nw}  \delta \hat  \mu_{L}^{\mk_{L-1}(i)} - \mu^{\mk_{L-1}(i)} \right|^p \middle| \mk_{L-1} \right] \right]^{\frac{1}{p}} \\
& \leq \left(\underset{1 \leq i_1< ... <  i_{\Nw} \leq \Ns}{\max} \Eb \left[\left|  \sum_{j=1}^{\Nw} \delta \hat \mu_{L}^{i_j} -  \mu^{i_j} \right|^p \right] \right)^{\frac{1}{p}} \\
& \leq \underset{1 \leq i_1< ... <  i_{\Nw} \leq \Ns}{\max} 
 \sum_{j=1}^{\Nw}  \Eb\left[\left| \delta \hat \mu_{L}^{i_j} -  \mu^{i_j} \right|^p\right]^{\frac{1}{p}}.
\end{split}
\end{equation*}
To discuss the first term in the right-hand side of \eqref{eq: decompo ES - hat ES avec hat mu}, the permutation error, let us first define ${\rm S}_q[A]$ as the collection of the $q$ smallest elements of a set $A\subset \Nb$. If $i\in \[{1,\Nw}\cap \Ik_{\ell-1}\setminus  \Ik_{\ell}$, then $i\in {\rm S}_{q_{\ell}}[\Ik_{\ell-1}]\setminus  \Ik_{\ell} $ and therefore there exists $k_{i}\in {\mathcal R_{\ell}}:=\Ik_{\ell}\setminus {\rm S}_{q_{\ell}}[\Ik_{\ell-1}]$. Thus, on the set $\{\{i_{1},\ldots,i_{J}\}=(\Ik_{\ell-1}\setminus \Ik_{\ell})\cap  \[{1,\Nw}\}$, one can define  
$
\kk_\ell(i_{1}):=\max {\mathcal R_{\ell}} 
$
and $\kk_\ell(i_{j+1}):=\max\{ k<\kk_\ell(i_{j})~:~k\in {\mathcal R_{\ell}}\}$ for $j+1\le J$. Note that 
\begin{align}\label{eq: inclusion i in Ikell}
 \{i\in   \Ik_{\ell-1}\setminus \Ik_{\ell}\}\subset \{ \hat \mu_{\ell}^{\kk_{\ell}(i)}>\hat \mu_{\ell}^{i}\}\;\mbox{ and }\;|{\mathcal R_{\ell}}|\le q_{\ell-1}-q_{\ell},
\end{align}
since $ {\mathcal R_{\ell}} \subset \Ik_{\ell-1}\setminus {\rm S}_{q_{\ell}}[\Ik_{\ell-1}]$ and $|\Ik_{\ell-1}|= q_{\ell-1}$. Let ${\mathbf A}_{q, q'}$ denote the collection of subsets $A$ of $\[{q + 1,\Ns}$ such that $|A| = q'$. 
Then, it follows from \eqref{eq: ordre mu}, H\"older's inequality and \eqref{eq: inclusion i in Ikell} that 

\begin{equation*}
\begin{split}
\Eb \left[ \left| \frac{1}{\Nw} \sum_{i \leq \Nw} \mu^i - \mu^{\mk_{L-1}(i)} \right|^p \right]^{\frac{1}{p}}
&\leq \frac{1}{\Nw} \sum_{i \leq \Nw} \sum_{\ell=1}^{L - 1} \Eb \left[ \left| ( \mu^{i} - \mu^{\kk_{\ell}(i)})\1_{\{i\in \Ik_{\ell-1}\setminus \Ik_{\ell}\}} \right|^p \right]^{\frac{1}{p}} \\
&\leq \underset{i \leq \Nw}{\max} \sum_{\ell=1}^{L-1} \Eb \left[ \left| ( \mu^{i} - \mu^{\kk_{\ell}(i)})\1_{\{i\in \Ik_{\ell-1}\setminus \Ik_{\ell}\}} \right|^p \right]^{\frac{1}{p}} \\
&\leq \underset{i \leq \Nw}{\max} \sum_{\ell = 1}^{L - 1} \left( \max_{A\subset {\mathbf A}_{q_{\ell}, \delta q_\ell}} \sum_{k\in A}\Eb\left[ \left| ( \mu^{i}-   \mu^{k}) \right|^p\1_{\{i\in  \Ik_{\ell-1}\setminus \Ik_{\ell},\kk_{\ell}(i)=k\}} \right] \right)^{\frac{1}{p}} \\
&\leq \sum_{\ell = 1}^{L - 1} \left(\delta q_\ell \right)^{\frac{1}{p}} \max_{(i,k)\in [\![1, {\Nw}]\!]\times [\![q_{\ell}+1, {\Ns}]\!]}  (\mu^{i}-\mu^{k})\Pb[\hat \mu_{\ell}^{k}>\hat \mu_{\ell}^{i}]^{\frac{1}{p}}.
\end{split}
\end{equation*}

\end{proof}

\subsection{Error bound for Sub-Gamma distributions}

To illustrate how the general error bound of Proposition \ref{prop: borne Lp generale} can be used in practice to decide of the sequence $(q_{\ell},N_{\ell})_{\ell}$, we now consider the case where the components of $P_{|{\rm s}}$ have sub-gamma distributions, and apply Bernstein's inequality in \eqref{eq: borne Lp generale}, see e.g.~\cite[Chapter 2]{bercu2015concentration}. This requires the following assumption. 
\vspace{2mm}

\begin{assumption}\label{hyp:sub gamma} There exists $c\in \Rb_{+}$ such that the random variables $Z[i,k]:=(P_{|{\rm s}^{i}}-\mu^{i})-(P_{|{\rm s}^{k}}-\mu^{k})$, $i,k\le \Ns$, satisfy Bernstein's condition : 
\begin{equation*}
    \mathbb{E} \left[ \left| Z[i,k] \right|^p\right] 
    \leq \frac{p!\;c^{p-2}}{2} \mathbb{E} \left[   Z[i,k]^2\right],\;i,k\le \Ns, \;\mbox{ for all } p\ge 3.
\end{equation*}
\end{assumption}
From now on, we shall assume that the constant $c$ is known. It can usually be estimated in practice.

\begin{corollary}\label{cor: borne si sub gamma}
Assume that Assumption \ref{hyp:sub gamma} holds. Then, for all $p\ge 1$, 

\begin{align}\label{eq: borne Lp sub gamma}
    \mathbb{E}\left[ \left|\ES - \ESh\right|^p \right]^{\frac{1}{p}} \leq & \Frm_{p}(q,\Nmc)
\end{align}
in which 
\begin{align}
\Frm_{p}(q,\Nmc):= &
\sum_{\ell=1}^{L-1} (\delta q_{\ell})^{\frac{1}{p}} \max_{(i,k)\in [\![1, {\Nw}]\!]\times [\![q_{\ell}+1, {\Ns}]\!]}  (\mu^{i}-\mu^{k}) e^{-\frac{\Nmc_{\ell}(\mu^{i}-\mu^{k})^{2}}{ 2p(\sigma_{ik}^{2}+c(\mu^{i}-\mu^{k}))}}\nonumber \\
&+
  \frac{1}{\Nw} \frac{\delta N_L}{N_L} \underset{ 1 \leq i_1 <  ... < i_{{ {\Nw}}} \leq  {\Ns}  }{\max} \sum_{j=1}^{ {\Nw}} \left( C_{p, \sigma} \frac{p \sigma_{i_j}^p}{(\delta \Nmc_{L})^{\frac{p}{2}}} + C_{p, c} \frac{ p c^{p}}{(\delta \Nmc_{L})^{p}} \right)^{\frac{1}{p}}   \label{eq: def F} \\ 
  &+ \frac{1}{\Nw} \frac{N_{L-1}}{N_L} \sum_{i=1}^{ {\Ns}} \left( C_{p, \sigma} \frac{p \sigma_{i}^{p}}{(\Nmc_{L-1})^{\frac{p}{2}}}  + C_{p, c} \frac{p c^{p}}{(\Nmc_{L-1})^{p}} \right)^{\frac{1}{p}} \nonumber 
\end{align}
with 
\begin{equation}
\begin{split}
\begin{cases}
\sigma_{ik}^{2}:={\rm Var}[P_{|{\rm s}^{i}}-P_{|{\rm s}^{k}}] \;\mbox{ and }\; \sigma_{i}^{2}:={\rm Var}[P_{|{\rm s}^{i}}],\;i,k\le \Ns \\
C_{p, \sigma} := 2^{p-1} \Gamma \left(\frac{p}{2}\right) \; \mbox{ and } \;
C_{p, c} := 4^p \Gamma \left( p \right)
\end{cases}
\end{split}\label{eq: def Cpc Cpsig}
\end{equation}
where $\Gamma$ is the Gamma function defined by
\begin{equation*}
\Gamma\left(y\right) = \int_0^{+\infty} x^{y - 1} e^{- x} dx,\;y> 0.
\end{equation*}
\end{corollary}

The upper-bound of Corollary \ref{cor: borne si sub gamma} has two advantages on Proposition \ref{prop: borne Lp generale}. First, the dependence on $(q_{\ell},N_{\ell})_{\ell\ge 0}$ is more explicit. It depends on unknown quantities, but we can estimate (at least rough) confidence intervals for them, see e.g.~Section \ref{subsec: DPP determinist} below. Second, as we will see in the next section, it allows one to define a tractable deterministic optimal control problem satisfying a dynamic programming principle, or even simple heuristics (see Section \ref{sec: 2 Level}), to select an appropriate sequence $(q_{\ell},N_{\ell})_{\ell\ge 0}$.\\

 \begin{proof}[Proof of Corollary \ref{cor: borne si sub gamma}] The first term in \eqref{eq: def F} is an upper-bound for the first term in the right-hand side of \eqref{eq: borne Lp generale}, see \cite[Theorem 2.1]{bercu2015concentration}. As for the two other terms in \eqref{eq: borne Lp generale}, we use the usual argument, for $i\le \Ns$,
$$
 \Eb\left[\left| \delta \hat \mu_{L}^{i}-   \mu^{i} \right|^p\right]=\int_0^\infty p x^{p-1} \Pb[|\delta \hat \mu_{L}^{i}-   \mu^{i}|\ge x]dx 
$$
and 
$$
 \Eb\left[\left|  \hat \mu_{L - 1}^{i}-   \mu^{i} \right|^p\right]=\int_0^\infty p x^{p-1} \Pb[|\hat \mu_{L - 1}^{i}-   \mu^{i}|\ge x]dx ,
$$
and then appeal to \cite[Theorem 2.1]{bercu2015concentration} again to deduce that 
\begin{align*}
 \Eb\left[\left| \delta \hat \mu_{L}^{i}-   \mu^{i} \right|^p\right]\le& \int_0^\infty p x^{p-1} e^{-\frac{\delta \Nmc_{L} x^{2}}{ 2(\sigma_{i}^{2}+cx) } } dx\\
 \le&\int_0^\infty p x^{p-1} e^{-\frac{\delta \Nmc_{L} x^{2}}{ 4 \sigma_{i}^{2}}} \1_{\{x\le \frac{\sigma_{i}^{2}}{c}\}} dx + \int_0^\infty p x^{p-1} e^{-\frac{\delta \Nmc_{L} x}{ 4 c } }  \1_{\{x>\frac{\sigma_{i}^{2}}{c}\}}dx\\ 
 \leq& \frac{p(\sigma^{2}_{i})^{\frac{p}2}}{(\delta \Nmc_{L})^{\frac{p}{2}}} \int_0^\infty y^{p-1} e^{-\frac{y^{2}}{4}}dy
 +\frac{p c^{p}}{(\delta \Nmc_{L})^{p}} \int_0^\infty y^{p-1} e^{-\frac{y}{4}}dy, \\
 \leq& \frac{p \sigma_i^p}{(\delta N_L)^{\frac{p}{2}}} 2^{p-1} \Gamma \left( \frac{p}{2} \right) + \frac{p c^p}{(\delta N_L)^p} 4^p \Gamma(p),
\end{align*}
and
\begin{align*}
 \Eb\left[\left| \hat \mu_{L - 1}^{i} - \mu^{i} \right|^p\right]\le& \int_0^\infty p x^{p-1} e^{-\frac{\Nmc_{L - 1} x^{2}}{ 2(\sigma_{i}^{2}+cx)}} dx\\
 \le&\int_0^\infty p x^{p-1} e^{-\frac{\Nmc_{L - 1} x^{2}}{ 4 \sigma_{i}^{2}}} \1_{\{x\le \frac{\sigma_{i}^{2}}{c}\}} dx + \int_0^\infty p x^{p-1} e^{-\frac{\Nmc_{L-1} x}{ 4 c } }  \1_{\{x>\frac{\sigma_{i}^{2}}{c}\}}dx\\ 
 \leq& \frac{p(\sigma^{2}_{i})^{\frac{p}2}}{(\Nmc_{L - 1})^{\frac{p}{2}}} \int_0^\infty y^{p-1} e^{-\frac{y^{2}}{4}}dy
 + \frac{p c^{p}}{(\Nmc_{L - 1})^{p}} \int_0^\infty y^{p-1} e^{-\frac{y}{4}}dy, \\
\leq& \frac{p \sigma_i^{p}}{(N_{L - 1})^{\frac{p}{2}}} 2^{p-1} \Gamma \left(\frac{p}{2} \right) + 
\frac{p c^p}{(N_{L - 1})^p} 4^p \Gamma \left( p \right). 
\end{align*}
\end{proof}
 
\begin{remark}\label{rem: c=0} {If the $(\hat \mu^{i}_{\ell})_{i\le \Ns}$ and $(\delta \hat \mu^{i}_{\ell})_{i\le \Ns}$ are Gaussian, which is the case asymptotically, then the bound of Corollary \ref{cor: borne si sub gamma} remains valid with $c=0$. This fact will be used later on for simplifying our numerical algorithms.} 
\end{remark}
 
\subsection{Optimal a-priori allocation by deterministic dynamic programming based on fixed a-priori bounds}\label{subsec: DPP determinist}

Given $\Nmc:=(\Nmc_{\ell})_{0 \le \ell\le L}$ and $q=(q_{\ell})_{0 \le \ell\le L-1}$, the total computation cost is 
$$
{\rm C}(q,\Nmc):=\sum_{\ell=0}^{L-1} q_{\ell}(\Nmc_{\ell+1}-  \Nmc_{\ell})
$$ 
with the convention $\Nmc_{0}:=0$. Let $\Nc$ denote the collection of non-decreasing sequences $\Nmc:=(\Nmc_{\ell})_{0\le \ell\le L}$ with values in $\Nb$ such that $N_{0}=0$, and let $\Qc$ denote the collections of non-increasing sequences\footnote{We write $(q_{\ell})_{0\le \ell\le L}$ for convenience also $q_{L}$ will never play any role.} $q=(q_{\ell})_{0\le \ell\le L}$ with values in $\[{\Nw,\Ns}$ satisfying \eqref{eq; cond sur q}. In this section, we fix a total effort $K>0$ and recall how $\Frm_{p}(q,\Nmc)$, as defined in \eqref{eq: def F}, can be minimized 
over the collection $\Ac$ of sequences $(\Nmc,q)\in \Nc\times \Qc$ satisfying ${\rm C}(\Nmc,q)\le K$ by using a standard dynamic programming approach. 
\vs2

 Given $(\bar q,\bar N) \in  \Qc\times \Nc$ and $ {0\le }\ell  {\le }  {L-1}$, we write
\begin{align}\label{eq:optimal_a_priori_bound_determinist_algorithm}
\Frm_{p}(\ell,\bar q,\bar N):= &
\frac{1}{\Nw} \frac{\delta \bar N_L}{\bar N_L} \underset{ 1 \leq i_1 <  ... < i_{ {\Nw}} \leq  {\Ns}  }{\max} \sum_{j=1}^{ {\Nw}} \left( C_{p, \sigma} \frac{p \sigma^{p}_{i_j}}{(\delta \bar \Nmc_{L})^{\frac{p}{2}}} + C_{p, c} \frac{p c^{p}}{(\delta \bar \Nmc_{L})^p} \right)^{\frac{1}{p}} \nonumber\\ 
&+ \frac{1}{\Nw} \frac{\bar N_{L-1}}{\bar N_{L}}  \sum_{i=1}^{ {\Ns}} \left( C_{p, \sigma} \frac{p \sigma^{p}_{i}}{(\bar \Nmc_{L - 1})^{\frac{p}{2}}} + C_{p, c} \frac{p c^{p}}{(\bar \Nmc_{L - 1})^{p}} \right)^{\frac{1}{p}} + \1_{\{\ell<L-1\}}\sum_{\ell'=\ell+1}^{L-1} f_{p}(\bar q_{\ell'},\bar q_{\ell'-1},\bar N_{\ell'}), 
\end{align}
where
$$
f_{p}(\bar q_{\ell'},\bar q_{\ell'-1},\bar N_{\ell'}):=  (\delta \bar q_{\ell'})^{\frac{1}{p}} \max_{(i,k)\in [\![1,\Nw]\!]\times [\![\bar q_{\ell'}+1,\Ns]\!]}  (\mu^{i}-\mu^{k})
      e^{-\frac{\bar N_{\ell'}(\mu^{i}-\mu^{k})^{2}}{ 2p(\sigma_{ik}^{2}+c(\mu^{i}-\mu^{k})) } }, 
$$
and define 
\begin{align*}
\hat \Frm_{p}(\ell,\bar q,\bar N)=\min_{(\bar q',\bar N')\in \Ac(\ell,\bar q,\bar N)} \Frm_{p}(\ell ,\bar q',\bar N') 
\end{align*}
where\footnote{In the following, we only write {$\Ac(0)$} for {$\ell=0$} as it does not depend on $(\bar q, \bar N)$.}
$$
\Ac(\ell,\bar q,\bar N):=\{(\bar q',\bar N')\in \Qc\times \Nc : (\bar q'_{l},\bar N'_{l})_{ {0} \le l\le \ell}= (\bar q_{l},\bar N_{l})_{ {0} \le l\le \ell}\mbox{ and } \;{\rm C}(\bar q',\bar N')\le K\} \;,\; \ell\ge {0}.
$$
 Then, the dynamic programming principle implies that 
 \begin{align*}
\hat  \Frm_{p}(\ell,\bar q,\bar N)=& \min_{(\bar q',\bar N')\in \Ac(\ell,\bar q,\bar N)}\left[\hat \Frm_{p}(\ell+1,\bar q',\bar N')+f_{p}(\bar q'_{\ell+1},\bar q_{\ell},\bar N'_{\ell+1})\right],\; \mbox{ for }0\le \ell< {L-1}.
 \end{align*}
 This reduces the search for an optimal selection of $(\bar \Nmc,\bar q)$ to $L-1$ one-step optimization problems, which is much simpler to solve than the optimization problem associated to the left-hand side of \eqref{eq: borne Lp sub gamma}. 
 \vs2
 
 In practice, the exact values of $(\mu^{i},\sigma^{2}_{i})_{i\le \Ns}$ and $(\sigma^{2}_{ik})_{i,k\le \Ns}$ are not known. However, one can consider robust versions of the above. For instance, if we know that there exists some $(\underline{\delta_{q}},\overline{\delta_{q}})_{q\le \Ns}$ and $\overline \sigma^{2}$ such that 
 \begin{eqnarray}\label{eq: a prior intervals for mu sig}
\left\{
\begin{array}{c}
0\le \underline{\delta_{q}}\le \mu^{i}-\mu^{k}\le \overline{\delta_{q}},\; (i,k)\in  [\![1,\Nw]\!]\times [\![q+1,\Ns]\!] \\
\sigma_{i}^{2} \vee\sigma_{k}^{2} \vee \sigma_{ik}^{2}\le \overline \sigma^{2},\;(i,k)\in  [\![1,\Nw]\!]\times [\![\Nw+1,\Ns]\!], 
\end{array}\right.\label{eq: bound mu sigma}
\end{eqnarray}
then one can similarly minimize the upper-bound of $\Frm_{p}$ defined as 
\begin{align*}
&
\frac{\delta \bar N_L }{\bar N_L} \left( C_{p, \sigma} \frac{p \overline \sigma^{p}}{(\delta \bar \Nmc_{L})^{\frac{p}{2}}} + C_{p, c} \frac{p c^{p}}{(\delta \bar \Nmc_{L})^{p}} \right) + \frac{\Ns}{\Nw} \frac{\bar N_{L-1}}{\bar N_L} \left( C_{p, \sigma} \frac{p \overline \sigma^{p}}{(\bar \Nmc_{L - 1})^{\frac{p}{2}}} + C_{p, c} \frac{p c^{p}}{(\bar \Nmc_{L - 1})^{p}} \right) \\
&+ \1_{\{\ell<L-1\}}\sum_{\ell'=\ell }^{L-1} \tilde f_{p}(\bar q_{\ell'},\bar q_{\ell'-1},\bar N_{\ell'}), 
\end{align*}
 with 
 $$
 \tilde f_{p}(\bar q_{\ell'},\bar q_{\ell'-1},\bar N_{\ell'}):=  
 (\delta\bar q_{\ell'})^{\frac{1}{p}} \max_{\underline{\delta}_{{\overline q_{\ell'}}}\le \delta\le \overline{\delta}_{{\overline q}_{\ell'}}}\;\delta
      e^{-\frac{\bar N_{\ell'}\delta^{2}}{ 2p(\overline \sigma^{2}+c\delta) }}.
 $$ 
 This corresponds to a worst case scenario, when only the a priori bounds $(\underline{\delta_{q}},\overline{\delta_{q}})_{q\le \Ns}$ and $\overline \sigma^{2}$ are known. 
 In the above, one can also impose that $q$ takes values in a given subset of $Q$ of $\Qc$. In this case, we will only need to know $(\underline{\delta_{q}},\overline{\delta_{q}})_{q\in \bar Q}$. 
\vs2

We refer to Section \ref{sec: tests numeriques} below for numerical tests that show that such an algorithm seems to perform pretty well. Note that the optimization can be done off-line. 

\subsection{Simplified 2-levels algorithm for a linear indifference zone's size}\label{sec: 2 Level}

Inspired by \cite{bahadur1950problem,bechhofer1954single,bechhofer1954tow,gupta1991sequential}, we assume here that we know the value of a constant $\delta_{0}>0$ such that the impacts of the $\Nw$ worst scenarios have values that are separated by at least $(k - \Nw) \delta_0$ from the $k$-th worst scenario, for $k>\Nw$: 
\begin{equation}\label{assumption:5}
   \mu^{ {\Nw}} - \mu^{ {k}} \geq \left(k - \Nw \right) \delta_0,\; \forall\; k \in [\![\Nw + 1, \Ns]\!].
\end{equation}

To illustrate this, we plot on Figures \ref{Fig: delta0 1}-\ref{Fig: delta0 4} the curves $k\mapsto  |\mu^{ {\Nw}} - \mu^{ {k}}|$ for different formerly used test books of Natixis. 
We see that they are more flat on the interval $[100,120]$, so that a rather conservative value would be the minimum (over the different books) of $(\mu^{100} - \mu^{120})/20$. {Another choice in practice could be to take the ratio $ (\mu^{\Nw} - \mu^{100})/(100 - \Nw)$ which amounts to considering only the first part of the curve, and neglecting points that are anyway far from the worst scenarios.}
\begin{figure}[H]
\begin{minipage}[t]{0.5\linewidth}
\centering
\includegraphics[width=\linewidth]{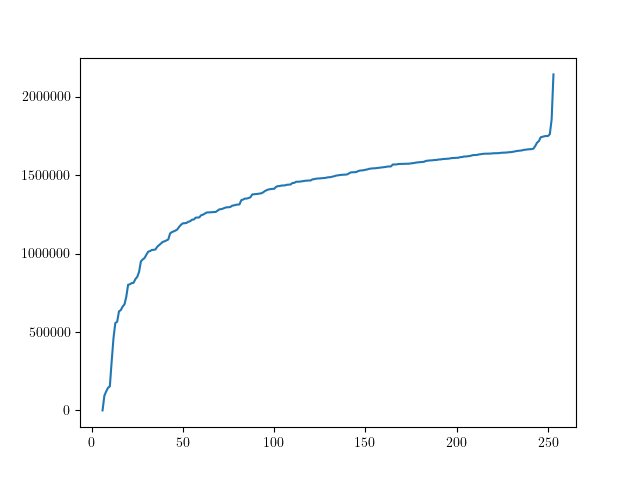}
\caption{ $k\mapsto  |\mu^{ {\Nw}} - \mu^{ {k}}|$ for   Book \#1}\label{Fig: delta0 1}
\end{minipage}
\begin{minipage}[t]{0.5\linewidth}
\centering
\includegraphics[width=\linewidth]{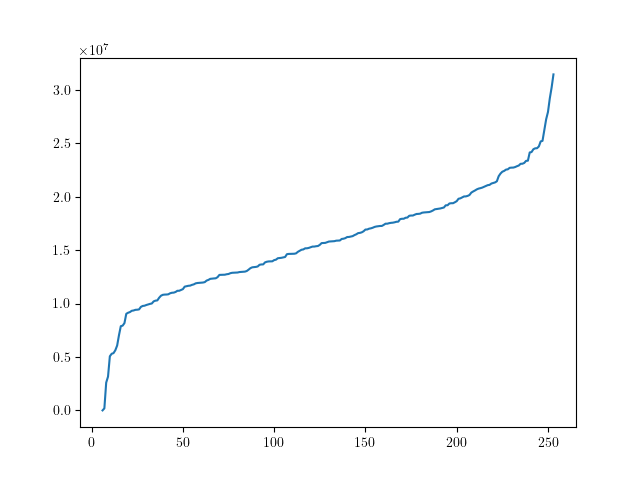}
\caption{ $k\mapsto  |\mu^{ {\Nw}} - \mu^{ {k}}|$ for   Book \#2}\label{Fig: delta0 2}
\end{minipage}
\end{figure}
\begin{figure}[H]
\begin{minipage}[t]{0.5\linewidth}
\centering
\includegraphics[width=\linewidth]{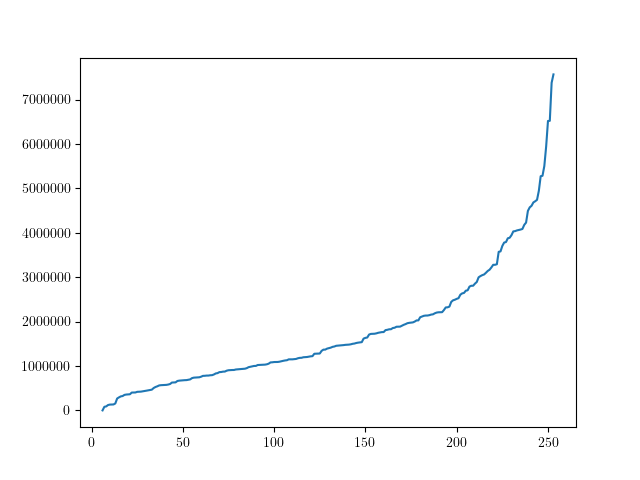}
\caption{ $k\mapsto  |\mu^{ {\Nw}} - \mu^{ {k}}|$ for   Book \#3}\label{Fig: delta0 3}
\end{minipage}
\begin{minipage}[t]{0.5\linewidth}
\centering
\includegraphics[width=\linewidth]{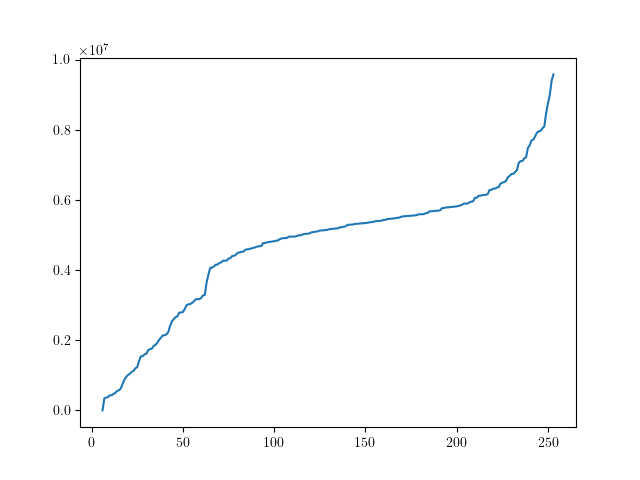}
\caption{ $k\mapsto  |\mu^{ {\Nw}} - \mu^{ {k}}|$ for   Book \#4}\label{Fig: delta0 4}
\end{minipage}
\end{figure}

We now consider a simplified version of the algorithm of Section \ref{subsec: deterministic algorithm} where we only do one intermediate ``fast pricing'' (meaning $N_{1}$ rather small) and one final ``full pricing" (meaning $N_{2}$ large). In theory, this corresponds to $L=3$ with $q_{2}
=\Nw$, $\delta N_{3}=0$ and $\delta N_{2}\to \infty$. As $\delta N_{2}\to \infty$, the second and third terms in \eqref{eq: def F} vanish, as well as the component of the first term corresponding to $\ell=2$. We therefore neglect them. In practice, we only take $N_{2}$ large enough (and given) from the point of view of the bank, and minimize over $(q_{1},N_{1})$ the remaining term in \eqref{eq: def F}:
\begin{align*} \label{simplified_2_steps_error} 
F_{1}^{\infty}(q_{1}):= (\Ns - q_{1})^{\frac{1}{p}} \max_{(i,k)\in [\![1,\Nw]\!]\times [\![q_{1}+1,\Ns]\!]}  (\mu^{i}-\mu^{k}) e^{-\frac{\Nmc_{1}(\mu^{i}-\mu^{k})^{2}}{ 2p(\overline \sigma^{2}+c(\mu^{i}-\mu^{k})) } },  
\end{align*}
in which $\bar \sigma$ is estimated to be as in \eqref{eq: a prior intervals for mu sig}, 
under the computation cost constraint 
$$
{\rm C}\left(N_1, q_1 \right) = q_1 (N_2 - N_1) + \Ns N_1 \le K
$$
for some given maximal cost $K\in {\mathbb N}^{*}$.  

{For $N_{1}$ (or $K$) large enough, the condition \eqref{assumption:5} leads} to minimizing over $q_1 \in [\![\Nw,\Ns]\!]\cap [1,K/N_{2}]$ {the upper-bound}
\begin{equation}\label{simplified_2_steps_optimization}
  h^{p}_0(q_1):=  (\Ns-q_1)^{{\frac1p}} \times 
(q_1+1-\Nw) \delta_0 \exp \left(- \frac{\left( K - q_1 N_2 \right) (q_1+1-\Nw)^2 \delta_0^2}{2  p (\Ns  {-q_{1}}) \left( \overline \sigma^{2} + c \left(q_1+1-\Nw\right) \delta_0 \right)} \right).
\end{equation}

 The optimal $q_{1}^{*}$ can then be found easily by performing a one-dimensional numerical minimization. {Upon replacing $\Ns  {-q_{1}}$ by $\Ns$ in the denominator of the exponential term, which provides a further upper-bound, the optimum can even be computed explicitly, see Appendix \ref{sec: h1 explicit}.} This provides a very easy to use algorithm.

Considering the case $p=1$, let us now perform first numerical tests to see if the proxy based on $h_{0}^{1}$ is far from $F^{\infty}_1$. We use the parameters of Tables \ref{tab:sample_books_params} and \ref{tab:computing_power} below and $\mu^{i}=-i\delta_{0}$, $i\le \Ns$, where $\delta_{0}:= (\mu^{\Nw} - \mu^{100})/(100 - \Nw)$ for the $\mu^i$s of Figure \ref{fig:book_sample_S1_raw_mu}. In particular, we take $c=0$, see Remark \ref{rem: c=0}.
 
In Figure \ref{fig: comp F et q 2-level bis}, the two increasing curves show the optimum $q_1^*$ (right axis) as found when applying the deterministic dynamic programming algorithm (dashed line) of Section \ref{subsec: DPP determinist} associated to the real sample book curve of Figure \ref{fig:book_sample_S1_raw_mu}, and the heuristic (solid line) based on \eqref{simplified_2_steps_optimization}. The two decreasing curves show the corresponding $F^{\infty}_1(q_1^*)$ (left axis) found when applying the deterministic dynamic programming algorithm (dotted line) and the heuristic (dashdot line). We see that the heuristic and the real minimizer are extremely close. {The noise in the lines associated to the dynamic programming algorithm are due to grid effects.}

\begin{figure}[H]
  \centering
  \includegraphics[width=0.6\linewidth]{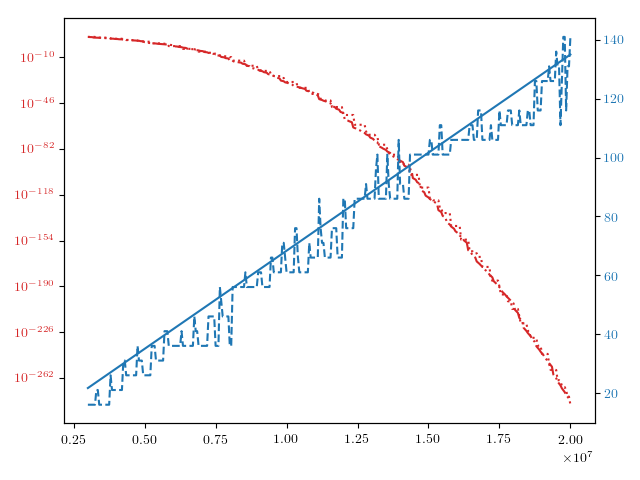}
  \caption{$q_1^*$ vs $K$ for the distribution of Figure \ref{fig:book_sample_S1_raw_mu}}\label{fig: comp F et q 2-level bis}
\end{figure}

\begin{table}[H]
\centering
\begin{tabular}{ |c|c|}
\hline
$\delta_0$ & 2 766\\
$c$ & 0 \\
$\bar \sigma$ & $\sqrt{2 (1 - \rho)} \times \, 2\,200\,000$ \\
$\rho $ & $0.6$ \\
$\Ns $ & 253 \\
$\Nw$ & 6 \\
\hline
\end{tabular}
\caption{ {Sample Book Parameters}}
\label{tab:sample_books_params}
\end{table}

\begin{table}[H]
\centering
\begin{tabular}{ |c|c|}
\hline
$K$ & $10^7$\\
$N_2$ & $10^5$ \\
\hline
\end{tabular}
\caption{ {Computing Power}}
\label{tab:computing_power} 
\end{table}

 \begin{figure}[H]
  \centering
  \includegraphics[width=\linewidth]{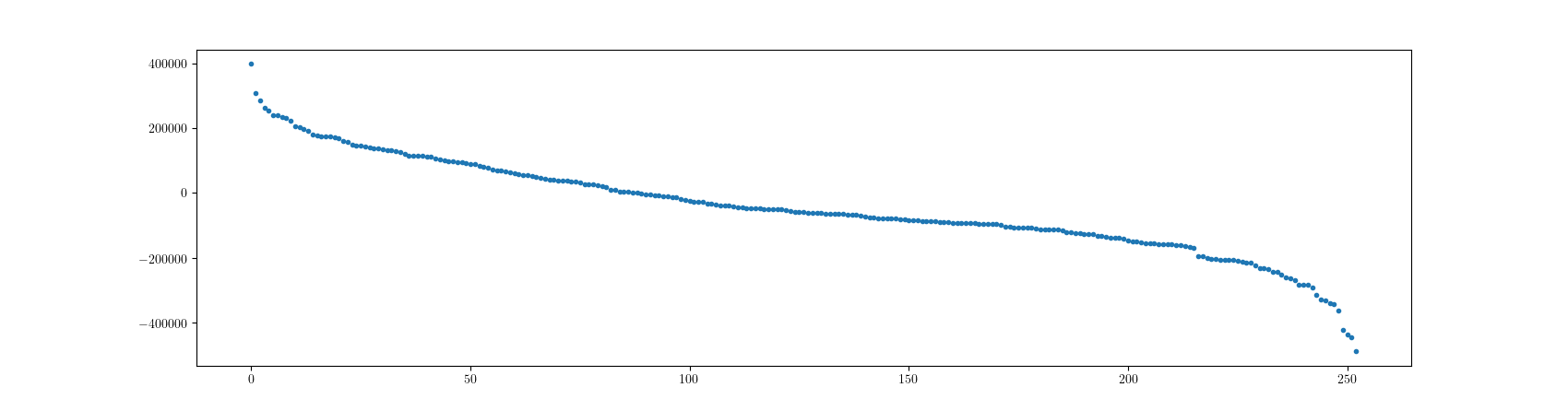}
  \caption{Sample book distribution : $i \mapsto \mu^i$ for $i\le \Ns$}
  \label{fig:book_sample_S1_raw_mu}
\end{figure}

\section{Adaptative algorithm}\label{sec : adaptative}

Although the true value $\theta_{\circ}=(\mu_{\circ},\Sigma_{\circ})$ of the vector of means and of the covariance matrix of $P_{|\rm s}$ are unknown, we can set on it a prior distribution, e.g.~based on previous Monte Carlo experiments, rather than just working on robust bounds as in the end of Section \ref{subsec: DPP determinist}. Since the estimation of $\ES$ uses Monte Carlo simulations of $P_{|\rm s}$, the knowledge of these quantities can be improved along the different steps $\ell$ of our estimation procedure. This suggests an adaptative algorithm for the optimization of the numerical effort allocation, in which we learn progressively the true value of these parameters, or part of them.
From now on, we therefore view the true value of the parameters as a random variable $\tilde \theta:=(\tilde \mu,\tilde \Sigma)$ on which a prior law $\nu_{0}$ is set. At each step $\ell$, new Monte Carlo simulations will allow us to update this prior, and our strategy for the next steps accordingly. 

\subsection{Error bounds and convergence for predictable strategies}

Let us first adapt the proof of Proposition \ref{prop: borne Lp generale} and Corollary \ref{cor: borne si sub gamma} to the case where controls are not deterministic but stochastic processes. Given a stochastic process $\alpha$ with values in $\Qc\times \Nc$, we set $(q^{\alpha},N^{\alpha}):=\alpha$ where $q^{\alpha}$ and $N^{\alpha}$ are respectively $\Qc$ and $\Nc$-valued. We then define $\hat \mu^{\alpha}=(\hat \mu^{\alpha}_{\ell})_{\ell \le L}$, $(\Ik^{\alpha}_{\ell},\mk^{\alpha}_{\ell})_{\ell \le L}$ as in Section \ref{subsec: deterministic algorithm} except that we see $\hat \mu^{\alpha}_{\ell}$ as a $q^{\alpha}_{\ell}$-dimensional random variables with entries given by $(\hat \mu^{\alpha,i}_{\ell})_{i\in \Ik^{\alpha}_{\ell}}$. We use the same convention for $\delta \hat \mu^{\alpha}_{\ell}$, recall \eqref{eq : def delta hat mu i}. We say that $\alpha$ is admissible if it is predictable with respect to $(\Fc^{\alpha}_{\ell})_{\ell \le L}$ in which $\Fc^{\alpha}_{0}$ is trivial and $\Fc^{\alpha}_{\ell}=\Fc^{\alpha}_{\ell-1}\vee \sigma(P_{j}^{i}, (i,j)\in \Ik^{\alpha}_{\ell}\times \[{1,N^{\alpha}_{\ell}} )$. We call $\Ac^{\rm ad}$ the collection of such processes. Then, one defines 
$$
\ESh^{\alpha}:=\frac{1}{\Nw} \sum_{i=1}^{\Nw} \hat \mu^{\alpha,\mk_{L - 1}^{\alpha}(i)}_{L},\;\alpha \in \Ac^{\rm ad}.
$$ 
The true value of the expected shortfall is now also written as a random variable
$$
\widetilde \ES:=\frac1\Nw \sum_{i=1}^{\Nw} \tilde \mu^{\tilde \mk(i)}, 
$$ 
in which $\tilde \mk$ is the random permutation such that 
\begin{equation*} 
\left\{
\begin{array} {l}
\tilde \mu^{\tilde \mk(1)} \ge  \tilde \mu^{\tilde \mk(2)} \ge \ldots \ge  \tilde \mu^{\tilde \mk(\Ns)},  
\\
\tilde \mk(i)<\tilde \mk(i')\;\mbox{ if } \tilde \mu^{\tilde \mk(i)} = \tilde \mu^{\tilde \mk(i')}\mbox{ for } 1\le i< i'\le \Ns.
\end{array}\right.
\end{equation*}

We let $\Mc$ be a collection of laws on $\Rb^{\Ns}\times \Sb^{\Ns}$, where $\Sb^{\Ns}$ denotes the collection of covariance matrices of size $\Ns$. Given $\nu \in \Mc$, we denote by $\Eb^{\nu}$ the expectation operator given that $\tilde \theta$ admits the law $\nu$. When $\nu$ is a Dirac mass, we retrieve the situation of Section \ref{sec: cas determinist} (up to re-ordering in a deterministic way the components of $\mu$).

We first provide a natural extension of Proposition \ref{prop: borne Lp generale}. 

\begin{proposition}{\label{prop: borne Lp generale controle aleatoire}} For all $p\ge 1$, $\nu \in \Mc$, and $\alpha\in\Ac^{\rm ad}$,  
\begin{align}
    \mathbb{E}^{\nu}\left[ \left|\widetilde \ES - \ESh^{\alpha}\right|^p \right]^{\frac{1}{p}} \leq & 
     \frac{1}{\Nw} \Eb^{\nu}\left[\left| \sum_{i\in \Ik_{L - 1}^{\alpha}} \hat \mu_{L}^{\alpha,i}-  \tilde \mu^{{i}} \right|^p\right]^{\frac{1}{p}} \label{eq: borne Lp generale controle aleatoire} \\
   +&  { \sum_{\ell=1}^{L - 1} \Eb^{\nu}\left[\delta q^{\alpha}_{\ell} \max_{(i,k)\in  \tilde \mk^{\alpha}_{\ell-1}(\[{1,\Nw}\times [\![q^{\alpha}_{\ell}+1,q^{\alpha}_{\ell-1}]\!])}  (\tilde \mu^{i}-\tilde \mu^{k})^p \Pb^{\nu}[\hat \mu_{\ell}^{\alpha,k}>\hat \mu_{\ell}^{\alpha,i} |\Fc^{\alpha}_{\ell-1}\vee \sigma(\tilde \theta)] \right]^{\frac{1}{p}}},\nonumber
\end{align}
with the convention $\max_{\emptyset}=0$ {and in which $\tilde \mk^{\alpha}_{\ell-1}$ is defined as $\tilde \mk$ but on the subset $\Ik^{\alpha}_{\ell-1}$ instead of $\Ik^{\alpha}_{0}=[\![1,\Ns]\!]$.}

\end{proposition}

 \begin{proof} We proceed as in the proof of Proposition \ref{prop: borne Lp generale} to obtain that 
 \begin{align*}
 \mathbb{E}^{\nu}\left[ \left|\widetilde\ES - \ESh^{\alpha}\right|^p \right]^{\frac{1}{p}} 
\le& \Eb^{\nu}\left[\left|\frac{1}{\Nw} \sum_{i\le \Nw} \hat \mu_{L}^{\alpha,\mk^{\alpha}_{L - 1}(i)}-  \tilde  \mu^{\mk^{\alpha}_{L - 1}(i)} \right|^p\right]^{\frac{1}{p}} +\Eb^{\nu}\left[\left|\frac{1}{\Nw} \sum_{i\le \Nw} \tilde \mu^{\tilde \mk(i)} -   \tilde \mu^{\mk^{\alpha}_{L - 1}(i)} \right|^p\right]^{\frac{1}{p}},
\end{align*}
where 
\begin{align*}
 \Eb^{\nu}\left[\left|\frac{1}{\Nw} \sum_{i\le \Nw}  \hat \mu_{L}^{\alpha,\mk^{\alpha}_{L - 1}(i)}-   \tilde \mu^{\mk^{\alpha}_{L - 1}(i)} \right|^p\right]^{\frac{1}{p}} 
 &= \frac{1}{\Nw}   \Eb^{\nu}\left[\left| \sum_{i\in \Ik_{L - 1}^{\alpha}} \hat \mu_{L}^{\alpha,i}-   \tilde \mu^{i} \right|^p\right]^{\frac{1}{p}} 
  .
\end{align*} 
We define $\kk^{\alpha}_{\ell}$ as $\kk_{\ell}$ in the proof of Proposition \ref{prop: borne Lp generale} for the strategy $\alpha$, with $ \mathcal R_{\ell}$ replaced by $\mathcal R^{\alpha}_{\ell}$ $:=$ $\Ik^{\alpha}_{\ell}\setminus \tilde \mk({\rm S}_{q^{\alpha}_{\ell}}[\tilde \mk^{-1}(\Ik_{\ell-1})])$. Then, 
\begin{align*}
&\Eb^{\nu}\left[\left|\frac{1}{\Nw} \sum_{i\le \Nw} \tilde \mu^{\tilde \mk(i)} -   \tilde \mu^{\mk^{\alpha}_{L - 1}(i)} \right|^p\right]^{\frac{1}{p}} \\
&\le \Eb^{\nu}\left[ \left|\frac{1}{\Nw} \sum_{i\le \Nw}\sum_{\ell=1}^{L - 1} ( \tilde \mu^{\tilde \mk(i)}-   \tilde \mu^{\kk^{\alpha}_{\ell}(\tilde \mk(i))})\1_{\{\tilde \mk(i)\in \Ik^{\alpha}_{\ell-1}\setminus \Ik^{\alpha}_{\ell}\}} \right|^p\right]^{\frac{1}{p}} \\
&\le \frac{1}{\Nw} \sum_{\ell=1}^{L - 1} \sum_{i\le \Nw} \Eb^{\nu}\left[\left| ( \tilde \mu^{\tilde \mk(i)}-  \tilde  \mu^{\kk^{\alpha}_{\ell}(\tilde \mk(i))}) \right|^p\1_{\{\tilde \mk(i)\in  \Ik^{\alpha}_{\ell-1}\setminus \Ik^{\alpha}_{\ell}\}} \right]^{\frac{1}{p}}\\
&\le {\frac{1}{\Nw} \sum_{\ell=1}^{L - 1} \sum_{i\le \Nw} \Eb^{\nu}\left[ \sum_{k\in \tilde \mk^{\alpha}_{\ell-1}([\![q^{\alpha}_{\ell}+1,q^{\alpha}_{\ell-1}]\!])}\Eb^{\nu}\left[ \left| ( \tilde \mu^{\tilde \mk(i)} - \tilde  \mu^{k}) \right|^p\1_{\{\tilde \mk(i)\in  \Ik^{\alpha}_{\ell-1}\setminus \Ik^{\alpha}_{\ell},\kk^{\alpha}_{\ell}(\tilde \mk(i))=k\}}|\Fc^{\alpha}_{\ell-1} \vee \sigma(\tilde \theta)\right]\right]^{\frac{1}{p}}} \\
& {\le \sum_{\ell=1}^{L - 1} \Eb^{\nu}\left[\delta q^{\alpha}_{\ell} \max_{(i,k)\in  \tilde \mk^{\alpha}_{\ell-1}(\[{1,\Nw}\times  [\![q^{\alpha}_{\ell}+1,q^{\alpha}_{\ell-1}]\!])}  (\tilde \mu^{i}-\tilde \mu^{k})^p \Pb^{\nu}[\hat \mu_{\ell}^{\alpha,k}>\hat \mu_{\ell}^{\alpha,i} |\Fc^{\alpha}_{\ell-1}\vee \sigma(\tilde \theta)] \right]^{\frac{1}{p}}.
}  \end{align*} 
   \end{proof}

 \begin{remark}\label{rem: borne Lp control aleatoire vs determinist} Note that, when $\alpha$ is deterministic and $\nu$ is concentrated on a Dirac, the right-hand side of \eqref{eq: borne Lp generale controle aleatoire} is bounded from above by 
 \begin{align*}
 &\frac{1}{\Nw} \frac{\delta N_L^{\alpha}}{N_{L}^{\alpha}} \underset{1 \leq i_i < ... < i_{\Nw} \leq \Ns}{\max} \sum_{j = 1}^{\Nw} \Eb^{\nu} \left[  \left| \delta \hat \mu^{\alpha,i_j}_{L}-   \tilde \mu^{i_j} \right|^p \right]^{\frac{1}{p}}
 + \frac{1}{\Nw} \frac{N_{L - 1}^\alpha}{N_{L}^{\alpha}} \sum_{i=1}^{\Ns} \Eb^{\nu} \left[ \left| \hat \mu^{\alpha,i}_{L - 1}-   \tilde \mu^{i} \right|^p \right]^{\frac{1}{p}} \\
 +
& \sum_{i=1}^{\Ns} \sum_{\ell=1}^{L - 1} \left(\delta q^{\alpha}_{\ell}\right)^{\frac{1}{p}} \Eb^{\nu}\left[\max_{(i,k)\in \tilde \mk\[{1,\Nw}\times [\![q^{\alpha}_{\ell}+1,\Ns]\!]}  (\tilde \mu^{i}-\tilde \mu^{k})^{p}\Pb^{\nu}[\hat \mu_{\ell}^{\alpha,k}>\hat \mu_{\ell}^{\alpha,i} |\Fc^{\alpha}_{\ell-1}\vee \sigma(\tilde \theta) ]\right]^{\frac{1}{p}},
\end{align*}
which coincides with the bound of Proposition \ref{prop: borne Lp generale} 
 \end{remark}
 
 The above guarantees the convergence of the algorithm.

\begin{proposition}\label{prop: convergence cas strat random} Let $(K^{n})_{n\ge 1}\subset {\mathbb N}^{*}$ be a sequence converging to infinity and let $(\alpha^{n})_{n\ge 1}$ be a sequence in $\Ac^{\rm ad}$ such that 
${\rm C}(q^{\alpha^{n}},N^{\alpha^{n}})\le K^{n}$ for each $n\ge 1$. Assume further that 
$
\min_{1\le \ell\le L} N^{\alpha^{n}}_{\ell}\to \infty$ {\rm a.s.} 
Let $\nu$ be concentrated on the Dirac mass on $\theta_{\circ}$. Then, 
$$
\mathbb{E}^{\nu}\left[ \left|\widetilde \ES - \ESh^{\alpha^{n}}\right|^p \right] \to 0\;\;\mbox{ as $n\to \infty$.}
$$ 
\end{proposition}

\proof It suffices to use the fact that, for some $C_{p}>0$, 
\begin{align*}
\Eb^{\nu}\left[\left| \hat \mu_{\ell}^{\alpha^{n},i} - \tilde \mu^{i} \right|^p\right]
&\le C_{p}\Eb^{\nu}\left[ \frac{\delta N^{\alpha^{n}}_{\ell}}{N^{\alpha^{n}}_{\ell}}  \Eb^{\nu}\left[\left|  \delta \hat \mu_{\ell}^{\alpha^{n},i} -    \mu^{i}_{\circ} \right|^p|\Fc^{\alpha^{n}}_{\ell-1}\right]\right]+
C_{p} \Eb^{\nu}\left[\frac{  N^{\alpha^{n}}_{\ell-1}}{N^{\alpha^{n}}_{\ell}}\left|   \hat \mu_{\ell-1}^{\alpha^{n},i} -    \mu^{i}_{\circ} \right|^p\right],
\end{align*}
in which 
$$
\frac{\delta N^{\alpha^{n}}_{\ell}}{N^{\alpha^{n}}_{\ell}} \Eb^{\nu}\left[\left|  \delta \hat \mu_{\ell}^{\alpha^{n},i} -    \mu^{i}_{\circ} \right|^p|\Fc^{\alpha^{n}}_{\ell-1}\right]\to 0,\;\nu_{\circ}-{\rm a.s.}, 
$$
for all $\ell>1$ and $i\le \Ns$. By induction, this implies that 
$$
\Eb^{\nu}\left[\left| \hat \mu_{\ell}^{\alpha^{n},i} - \tilde \mu^{i} \right|^p\right]=\Eb^{\nu}\left[\left| \hat \mu_{\ell}^{\alpha^{n},i} -   \mu^{i}_{\circ} \right|^p\right]\to 0
$$
for all $\ell\le L$ and $i\le \Ns$. Moreover, for some $C>0$, 
\begin{align*}
\Eb^{\nu}\left[  (\tilde \mu^{i}-\tilde \mu^{k})^p \Pb^{\nu}[\hat \mu_{\ell}^{\alpha^{n},k}>\hat \mu_{\ell}^{\alpha^{n},i} |\Fc^{\alpha}_{\ell-1}\vee \sigma(\tilde \theta)] \right]
&\le C\1_{\{\mu^{i}_{\circ}-\mu^{k}_{\circ}>0\}}\frac{\Eb^{\nu}[|\hat \mu_{\ell}^{\alpha^{n},i}-\hat \mu_{\ell}^{\alpha^{n},k}-(\mu^{i}_{\circ}-\mu^{k}_{\circ})|]}{\mu^{i}_{\circ}-\mu^{k}_{\circ}}\to 0
\end{align*}
for all $i<k$ and $\ell\le L-1$.
\endproof
 
 Using the fact that a control $\alpha \in \Ac^{\rm ad}$ is predictable, one can then proceed as in the proof of Corollary \ref{cor: borne si sub gamma} to derive a more tractable upper-bound. It appeals to the following version of Assumption \ref{hyp:sub gamma}.  
 
 \begin{assumption}\label{hyp:sub gamma control random} There exists $c>0$ such that, for all $\nu \in \Mc$,  
 \begin{equation*}
    \mathbb{E}^{\nu} \left[ \left| Z[i,k] \right|^p|\sigma(\tilde \theta)\right] 
    \leq \frac{p!\;c^{p-2}}{2} \mathbb{E}^{\nu} \left[   Z[i,k]^2|\sigma(\tilde \theta)\right]\;\nu-{\rm a.s.}, \;\mbox{ for all } \;i,k\le \Ns,\;p\ge 3.
\end{equation*}
\end{assumption}
 
\begin{corollary}\label{cor: borne si sub gamma controle aleatoire}
Let Assumption \ref{hyp:sub gamma control random} holds. Then, for all $p\ge 1$, $\alpha  \in \Ac^{\rm ad}$ and $\nu \in \Mc$, 

\begin{align*}\label{eq: borne Lp sub gamma controle aleatoire}
    \mathbb{E}^{\nu}\left[ \left|\widetilde \ES - \ESh^{\alpha}\right|^p \right]^{\frac{1}{p}} \leq &  \Frm^{\rm ad}_{p}(\alpha,\nu)
    \end{align*}
in which 
\begin{equation*}
\Frm^{\rm ad}_{p}(\alpha,\nu):= \frac{1}{\Nw}
\Eb^{\nu}\left[\left| \sum_{i\in \Ik_{L - 1}^{\alpha}} \hat \mu_{L}^{\alpha,i} - \tilde \mu^{i} \right|^p \right]^{\frac{1}{p}} 
+ \sum_{\ell=1}^{L - 1} \Eb^{\nu} \left[f^{\rm ad}_{p}(\ell,\alpha,\tilde \theta)\right]^{\frac{1}{p}} \label{eq: def F controle aleatoire}  
      \end{equation*}
where 
\begin{equation}\label{eq: def fap}
f^{\rm ad}_{p}(\ell,\alpha,\tilde \theta) := \delta q^{\alpha}_{\ell} \max_{(i,k)\in  {\tilde \mk^{\alpha}_{\ell-1}(\[{1,\Nw}\times [\![q^{\alpha}_{\ell}+1,q^{\alpha}_{\ell-1}]\!])}}   (\tilde \mu^{i}-\tilde \mu^{k})^{p}
      \left(e^{-\frac{\delta N^{\alpha}_{\ell}(\rho^{\alpha}_{\ell}[i,k])^{2}}{ 2(\tilde \sigma_{ik}^{2}+c \rho^{\alpha}_{\ell}[i,k]) } }\1_{\{\rho^{\alpha}_{\ell}[i,k]\ge  0\}}
+ \1_{\{\rho^{\alpha}_{\ell}[i,k]< 0\}}\right)
\end{equation}
with 
$$
\rho^{\alpha}_{\ell}[i,k]:=\tilde \mu^{i}- \tilde \mu^{k} + \frac{N^{\alpha}_{\ell-1}}{\delta N^{\alpha}_{\ell}}  
(\hat \mu_{\ell-1}^{\alpha,i}-\hat \mu_{\ell-1}^{\alpha,k})  \mbox{ for $\ell \ge 1$ and $i,k\le \Ns$.}
$$

\end{corollary}

\proof We use Bernstein's inequality, see \cite[Theorem 2.1]{bercu2015concentration}, conditionally to $\Fc^{\alpha}_{\ell-1}\vee \sigma(\tilde \theta)$, to deduce that 
\begin{align*}
&\Pb^{\nu}[\hat \mu_{\ell}^{\alpha,k}>\hat \mu_{\ell}^{\alpha,i} |\Fc^{\alpha}_{\ell-1}\vee \sigma(\tilde \theta) ]\\
&=\Pb^{\nu}[\delta \hat \mu_{\ell}^{\alpha,k}-\tilde \mu^{k}-(\delta \hat \mu_{\ell}^{\alpha,i}-\tilde \mu^{i})> \frac{N^{\alpha}_{\ell-1}}{\delta N^{\alpha}_{\ell}}  
(\hat \mu_{\ell-1}^{\alpha,i}-\hat \mu_{\ell-1}^{\alpha,k})-( \tilde \mu^{k}- \tilde \mu^{i}) |\Fc^{\alpha}_{\ell-1} \vee \sigma(\tilde \theta)]\\
&\le  e^{-\frac{\delta N^{\alpha}_{\ell}(\rho^{\alpha}_{\ell}[i,k])^{2}}{ 2(\tilde \sigma_{ik}^{2}+c \rho^{\alpha}_{\ell}[i,k]) } }\1_{\{\rho^{\alpha}_{\ell}[i,k]\ge  0\}}
+ \1_{\{\rho^{\alpha}_{\ell}[i,k]<  0\}}.
\end{align*}
\endproof

\subsection{A generic progressive learning algorithm}

Let us now describe how the result of Corollary \ref{cor: borne si sub gamma controle aleatoire} can be turned into a (stochastic) dynamic programming algorithm, in the spirit of Section \ref{subsec: DPP determinist}, that can be implemented in practice.
\vspace{5mm}

By Jensen's inequality, the upper-bound of Corollary \ref{cor: borne si sub gamma controle aleatoire} can be rewritten as
\begin{equation}\label{eq: bound Lp pour learning algo}
\Eb^\nu \left[ \left| \widetilde \ES - \ESh^{\alpha} \right|^p \right] \leq \Frm^{\rm ad}_{p}(0,\alpha,\nu)^p
\end{equation}
where
\begin{equation*} \label{equation:definition_upper_bound_control_problem}
\Frm^{\rm ad}_{p}(0, \alpha,\nu)  :=  \Eb^\nu \left[ \left| \frac{1}{\Nw} \sum_{i\in \Ik_{L - 1}^{\alpha}} \hat \mu_{L}^{\alpha,i} - \tilde \mu^{i} 
\right|^p + \sum_{\ell=1}^{L - 1} f^{\rm ad}_{p}(\ell,\alpha,\tilde \theta) \right],
\end{equation*}
to which we can associate the optimal control problem\footnote{Only the conditional law given $\Fc^{\alpha}_{\ell}$ of the components of $\tilde \theta$ corresponding to indexes in $\Ik^{\alpha}_{\ell}$ play a role in the definition of $\hat \Frm^{\rm ad}_{p}(\ell,\alpha,\nu)$ and $ \Frm^{\rm ad}_{p}(\ell,\alpha,\nu)$ below. To avoid introducing new complex notations, we shall indifferently take $\nu$ or only the conditional law of the corresponding components as an argument, depending on the context.}  
 \begin{align*}
\hat \Frm^{\rm ad}_{p}(\ell,\alpha,\nu)=\essinf{\alpha'\in \Ac^{\rm ad}(\ell,\alpha)}  \Frm^{\rm ad}_{p}(\ell ,  \alpha',\nu)\;\;\mbox{ for $0\le \ell \le L - 1$, $\nu \in \Mc$ and $\alpha\in \Ac^{\rm ad}$,}
\end{align*}
where 
$$
\Ac^{\rm ad}(\ell,\alpha):=\{\alpha'=( q', N')\in  \Ac^{\rm ad} : (\alpha'_{l})_{0\le l\le \ell}= (\alpha_{l})_{0\le l\le \ell}\mbox{ and} \;{\rm C}(q',N')\le K\}
$$
and
\begin{align*}
\Frm^{\rm ad}_{p}(\ell,\alpha',\nu)&:=  
\Eb^{\nu}\left[  \left| \frac{1}{\Nw} \sum_{i\in \Ik_{L - 1}^{\alpha'}} \hat \mu_{L}^{\alpha',i} - \tilde \mu^{i}  \right|^p + {\1_{\{\ell<L-1\}}} \sum_{l=\ell+1}^{L - 1} f^{\rm ad}_{p}(l,\alpha',\tilde \theta)~\Bigg|~\Fc^{\alpha'}_{\ell}\right] .
\end{align*}
  
 It admits a dynamic programming principle that involves a Bayesian update of the prior law on $\tilde \theta$ at each step of the algorithm, see e.g.~\cite{easley1988controlling}.

Let us first observe that, from step $\ell$ on, our bound only involves the components of $\tilde \theta$ associated to the indexes in $\Ik^{\alpha}_{\ell}$. We therefore set 
$$
\tilde \theta^{\alpha}_{\ell}=(\tilde \mu^{\alpha}_{\ell},\tilde \Sigma^{\alpha}_{\ell}):={\cal T}^{\Ik^{\alpha}_{\ell}}_{\Ik^{\alpha}_{\ell-1}}(\tilde \theta^{\alpha}_{\ell-1}),\;\ell\ge 1,\;\mbox{ with } \tilde \theta^{\alpha}_{0}:=\tilde \theta
$$ 
where, for two subsets $A'\subset A\subset \[{1,\Ns}$ and $(\mu,\Sigma)=( (  \mu^{i})_{i\in A},(  \Sigma^{ij})_{i,j\in A })$, we define 
$$
{\cal T}^{A'}_{A}(\mu,\Sigma)=(  (\mu^{i})_{i\in A'},(  \Sigma^{ij})_{i,j\in A' }).
$$ 
This means that the update of the prior can be restricted to a reduced number of components of $\tilde \theta$. This explains why we will concentrate on minimizing this upper-bound rather than directly the left-hand side of \eqref{eq: bound Lp pour learning algo}, which would lead to a very high-dimensional optimal control problem, at each step $\ell$. This way, we expect to reduce very significantly the computation cost of the corresponding ``optimal'' strategy.
\vs2

 In order to make the updating rule explicit, we use the following assumption. 
 
 \begin{assumption}\label{hyp: loi condi P sachant theta} Given $\nu_{0}\in \Mc$, there exists a measure $m$, such that, for all $\alpha\in \Ac^{\rm ad}$ and ${1}\le \ell\le L$, the law of $I^{\alpha}_{\ell}:=(P_{j}^{i}, (i,j)\in \Ik^{\alpha}_{{\ell-1}}\times \[{N^{\alpha}_{\ell-1}+1,N^{\alpha}_{\ell}})$ given $\Fc^{\alpha}_{\ell-1}\vee \sigma(\tilde \theta)$ admits $\nu_{0}$-a.s.~the density 
 $g^{\alpha}_{\ell}(\cdot,\tilde \theta^{\alpha}_{{\ell-1}}):=g(\cdot,\Ik^{\alpha}_{{\ell-1}},N^{\alpha}_{\ell-1},N^{\alpha}_{\ell},\tilde \theta^{\alpha}_{{\ell-1}})$ with respect to $m$, in which $g$ is a bounded measurable map\footnote{As for measurability, we identify $\Ik^{\alpha}_{{\ell-1}}$ to the element of $\Rb^{\Ns}$ with $i$-th entry given by $\1_{\{i\in \Ik^{\alpha}_{{\ell-1}}\}}$.}, that is continuous in its first argument, uniformly in the other ones. Moreover, for all $\alpha\in \Ac^{\rm ad}$ and $\ell\le L$, the law of $\tilde \theta$ given $\Fc^{\alpha}_{\ell}$ belongs to $\Mc$ $\nu_{0}$-a.s.
\end{assumption} 

Under this assumption, we can compute the law ${\nu^{\alpha,\ell-1}_{\ell}}$ of $\tilde \theta^{\alpha}_{{\ell-1}}={{\cal T}^{\alpha}_{\ell-1}(\tilde \theta)}$ given $\Fc^{\alpha}_{\ell}$ in terms of its counterpart $\nu^{\alpha}_{\ell-1}$ given $\Fc^{\alpha}_{\ell-1}$, 
in which 
$
{\cal T}^{\alpha}_{\ell-1} :={\cal T}^{\Ik^{\alpha}_{\ell-1}}_{\Ik^{\alpha}_{\ell-2}}\circ\cdots\circ {\cal T}^{\Ik^{\alpha}_{1}}_{\Ik^{\alpha}_{0}}.
$
It is given by 
\begin{equation*}\label{eq: nuell+1 en fonciton nuell}
{\nu^{\alpha,\ell-1}_{\ell}}=\Uc^{\frac12}(\ell,\alpha, \nu^{\alpha}_{\ell-1})
 \end{equation*}
 with $\nu^{\alpha}_{0}=\nu$ and 
\begin{align}\label{eq: def Uc}
\Uc^{\frac12}(\ell,\alpha, \nu^{\alpha}_{\ell-1})(A):=
\frac{\int_{{\cal D}^{\alpha}_{\ell-1}} g^{\alpha}_{\ell}(I^{\alpha}_{\ell}, { \theta}) 
\1_{\{{ \theta}\in A\}}
 {\nu^{\alpha}_{\ell-1}}  (d\theta)}{\int_{{\cal D}^{\alpha}_{\ell-1}} g^{\alpha}_{\ell}(I^{\alpha}_{\ell},  {\theta}) \nu^{\alpha}_{\ell-1}(d\theta)}
 \end{align}
for a Borel set $A$ of 
$
{\cal D}^{\alpha}_{\ell-1}:={\cal T}^{\alpha}_{\ell-1} (\Rb^{\Ns}\times \Sb^{\Ns}).
$ 
From this, one can deduce the law ${\nu^{\alpha}_{\ell}}$ of $\tilde \theta^{\alpha}_{ {\ell}}= {{\cal T}^{\alpha}_{\ell}(\tilde \theta)}$ given $\Fc^{\alpha}_{\ell}$, in the form 
$$
{\nu^{\alpha}_{\ell}}=\Uc(\ell,\alpha, \nu^{\alpha}_{\ell-1}),
$$ 
 by simply integrating on the components corresponding to indexes that are not in $\Ik^{\alpha}_{\ell}$ (meaning that $\Uc$ is explicit in terms of $\Uc^{\frac12}$).
\vs2

We are now in position to state our dynamic programming principle, see e.g.~\cite{easley1988controlling}. Again, note that {the law of} $f^{\rm ad}_{p}(\ell{+1},\alpha',\tilde \theta) $ {given $\Fc^{\alpha'}_{\ell}$} depends on $\tilde \theta$ only through $\tilde \theta^{\alpha'}_{\ell}$. For ease of notations, we identify all measures to an element of $\Mc$ (even if it supported by a space smaller than $\Rb^{\Ns}\times \Sb^{\Ns}$). 
\begin{proposition}\label{prop: dpp strat random} Let Assumption \ref{hyp: loi condi P sachant theta} hold. Then, for all $\alpha\in \Ac^{\rm ad}$, $0\le \ell\le L - 2$ and $\nu \in \Mc$, 
 \begin{align*}
\hat \Frm^{\rm ad}_{p}(\ell,\alpha,\nu)=\essinf{\alpha'\in \Ac^{\rm ad}(\ell,\alpha)}  \Eb^{\nu}[\hat \Frm^{\rm ad}_{p}(\ell+1 ,  \alpha',\Uc(\ell+1,\alpha', \nu)) + f^{\rm ad}_{p}(\ell+1,\alpha',\tilde \theta) | \Fc^{\alpha}_{\ell}].
\end{align*}
\end{proposition}

In principle, this dynamic programming algorithm allows one to estimate numerically the optimal policy $\alpha^{\star}$ in a feed-back form, off-line. Importantly, solving this problem given an initial prior $\nu_{0}$ is very different from first estimating the parameter $\tilde \theta$ and then solving the control problem as if $\tilde \theta$ was given. In the first case, we take into account the risk due to the uncertainly on the true value of $\tilde \theta$, not in the second one. 
\begin{remark}\label{rem : tau rho}
In practice, the algorithm requires estimating and manipulating the law of a high-dimensional parameter, at least at the first steps. But the above can be modified by changing the filtration $(\Fc^{\alpha}_{\ell})_{\ell \le L}$ in $(\bar \Fc^{\alpha}_{\ell})_{\ell \le L}$ with $\bar \Fc^{\alpha}_{\ell}$ $=$ $ \sigma(\1_{\ell \ge \tau^{\alpha}}P_{j}^{i}, (i,j)\in \Ik^{\alpha}_{\ell}\times \[{1,N^{\alpha}_{\ell}} )$ with $\tau^{\alpha}:=\inf\{l\le L: q^{\alpha}_{l}\le \rho\}$ for some $\rho>0$. In this case, no additional information is considered up to step $\tau^{\alpha}$, the update of the prior only takes place from step $\tau^{\alpha}$ on and it only concerns $\tilde \theta^{\alpha}_{\tau^{\alpha}}$ whose dimension is controlled by $\rho$. As for the first steps of the algorithm, namely before $\tau^{\rho}$, one can replace $f^{\rm ad}_{p}$ by a robust version in the spirit of Section \ref{subsec: DPP determinist}. 
\end{remark}
\begin{remark}\label{rem: prior GW}
The algorithm also requires knowing the conditional density $g^{\alpha}_{\ell}$. Although, $P_{|\sr}$ can be simulated, its conditional density is not known in general. However, one can use a proxy and/or again modify the flow of information to reduce to a more explicit situation. Let us consider the situation in which $(\Fc^{\alpha}_{\ell})_{\ell \le L}$ is replaced by $(\bar \Fc^{\alpha}_{\ell})_{\ell \le L}$ with $\bar \Fc^{\alpha}_{\ell}$ $=$ $\bar \Fc^{\alpha}_{\ell-1}\vee \sigma(\delta \hat \mu^{{\alpha},i}_{\ell}, i \in \Ik^{\alpha}_{\ell} )$ and $\bar \Fc^{\alpha}_{0}$ is trivial. Then, conditionally to $\bar \Fc^{\alpha}_{\ell-1}\vee \sigma(\tilde \theta^{\alpha}_{\ell-1})$, $\sqrt{\delta N^{\alpha}_{\ell}}(\tilde \Sigma^{\alpha}_{\ell})^{-1}\left(\delta \hat \mu^{\alpha}_{\ell}-\tilde \mu^{\alpha}_{\ell}\right)$ is asymptotically Gaussian as $\delta N^{\alpha}_{\ell}$ increases to infinity. In practice, we can do as if $\sqrt{\delta N^{\alpha}_{\ell}}(\tilde \Sigma^{\alpha}_{\ell})^{-1} \left(\delta \hat \mu^{\alpha}_{\ell}-\tilde \mu^{\alpha}_{\ell}\right)$ was actually following a standard Gaussian distribution, conditionally to $\tilde \theta^{\alpha}_{\ell}$ and $\Fc^{\alpha}_{\ell-1}$, which provides an explicit formula for the conditional density $\bar g^{\alpha}_{\ell}$ of $\delta \hat \mu^{\alpha}_{\ell}$ given $\tilde \theta^{\alpha}_{\ell-1}$ and $\Fc^{\alpha}_{\ell-1}$, to be plugged into \eqref{eq: def Uc}. Namely, the updating procedure takes the form 
$$
\nu^{\alpha}_{\ell}=\check \Uc(\ell,\alpha, \nu^{\alpha}_{\ell-1})
$$ 
where $\check \Uc$ is explicit.

Then, if the initial prior $\nu_{0}$ is such that $(\tilde \mu,\tilde \Sigma)$ is a Normal-inverse-Wishart distribution, all the posterior distribution $\nu^{\alpha}_{\ell}$, $\ell\le L$, are such that $(\tilde \mu,\tilde \Sigma)$ remains in the class of Normal-inverse-Wishart distributions with parameters that can computed explicitly from our simulations. Namely, if, given $\bar \Fc^{\alpha}_{\ell}$, $\tilde \Sigma$ has the distribution\footnote{Hereafter ${\cal W}^{-1}_{\is}(\Sigma)$ stands for the Inverse-Wishart distribution with degree of freedom $\is$ and scale matrix $\Sigma$, while ${\cal N}(\ms,\Sigma)$ is the Gaussian distribution with mean $\ms$ and covariance matrix $\Sigma$.} ${\cal W}_{\is^{\alpha}_{\ell}}^{-1}(\Sigma^{\alpha}_{\ell})$ and $\tilde \mu$ has the distribution ${\cal N}(\ms^{\alpha}_{\ell},\tilde \Sigma/\ks^{\alpha}_{\ell})$ given $\tilde \Sigma$, then the coefficients corresponding to the law given $\bar \Fc^{\alpha}_{\ell+1}$ are {\color{red}}
\begin{align}\label{eq: formule de transition}
&\left\{\begin{array}{rl}
\is^{\alpha}_{\ell+1}=&\is^{\alpha}_{\ell}+\delta N^{\alpha}_{\ell+1},\;\ks^{\alpha}_{\ell+1}=\ks^{\alpha}_{\ell}+\delta N^{\alpha}_{\ell+1},\;\ms^{\alpha}_{\ell+1}=\frac1{\kappa^{\alpha}_{\ell} + \delta N^{\alpha}_{\ell+1}} \left[\kappa^{\alpha}_{\ell} {\cal T}^{\Ik^{\alpha}_{\ell+1}}_{\Ik^{\alpha}_{\ell}} (\ms^{\alpha}_{\ell})+\delta N^{\alpha}_{\ell+1}\delta \hat \mu^{\alpha}_{\ell+1}\right]\\
\Sigma^{\alpha}_{\ell+1}=& {\cal T}^{\Ik^{\alpha}_{\ell+1}}_{\Ik^{\alpha}_{\ell}}(\Sigma^{\alpha}_{\ell})+\sum_{j=N^{\alpha}_{\ell}+1}^{N^{\alpha}_{\ell+1}} ({\cal T}^{\alpha}_{\ell+1}(P_{j})-\delta \hat \mu^{\alpha}_{\ell+1} )({\cal T}^{\alpha}_{\ell+1}(P_{j})-\delta \hat \mu^{\alpha}_{\ell+1})^{\top}\\
&+ \frac{\kappa^{\alpha}_{\ell} \delta N^{\alpha}_{\ell+1}}{\kappa^{\alpha}_{\ell} + \delta N^{\alpha}_{\ell+1}}({\cal T}^{\Ik^{\alpha}_{\ell+1}}_{\Ik^{\alpha}_{\ell}}(\ms^{\alpha}_{\ell})- \delta \hat \mu^{\alpha}_{\ell+1})({\cal T}^{\Ik^{\alpha}_{\ell+1}}_{\Ik^{\alpha}_{\ell}}(\ms^{\alpha}_{\ell})-\delta \hat \mu^{\alpha}_{\ell+1})^{\top} ,
\end{array}\right.\end{align}
see e.g.~\cite[Section 9]{NWishart}. Later on, we shall write the corresponding law as ${\cal NW}^{-1}(\ps^{\alpha}_{\ell+1})$ with 
$$
\ps^{\alpha}_{\ell+1}:=(\ms^{\alpha}_{\ell+1},\ks^{\alpha}_{\ell+1},\is^{\alpha}_{\ell+1},\Sigma^{\alpha}_{\ell+1}).
$$
\end{remark}

\subsection{Example of numerical implementation using neural networks}\label{subsec: neural network}
In this section, we aim at solving the version of the dynamic programming equation of Proposition \ref{prop: dpp strat random}, using an initial Normal-inverse-Wishart prior and the approximate updating procedure suggested in Remark \ref{rem: prior GW}:
 \begin{align*}
\check \Frm^{\rm ad}_{p}(\ell,\alpha,\nu)=\essinf{\alpha'\in \Ac^{\rm ad}(\ell,\alpha)} \Eb^{\nu}[\check \Frm^{\rm ad}_{p}(\ell+1 ,  \alpha',\check \Uc(\ell+1,\alpha', \nu)) + f^{\rm ad}_{p}(\ell+1,\alpha',\tilde \theta) | \bar \Fc^{\alpha}_{\ell} ],
\end{align*}
with $\check \Uc$ as in Remark \ref{rem: prior GW} and 
$$
\check \Frm^{\rm ad}_{p}(L - 1,\alpha,\nu):= \Eb^{\nu}\left[\left|\frac{1}{\Nw} \sum_{i\in \Ik_{L - 1}^{\alpha'}} \hat \mu_{L}^{\alpha',i}-   \tilde \mu^{i} \right|^p  ~\Bigg|~\bar \Fc^{\alpha}_{L - 1}\right] .
$$
It would be tempting to use a standard grid-based approximation. However, to turn this problem in a Markovian one, one needs to let the value function at step $\ell$ depend on $q^{\alpha}_{\ell}$, $N^{\alpha}_{\ell}$, $C^{\alpha}_{\ell}$, $\hat \mu^{\alpha}_{\ell}$ and $\ps^{\alpha}_{\ell}$, where $C^{\alpha}_{\ell}$ is the running cost of strategy $\alpha$ up to level $\ell$, defined for $\ell \neq 0$ by $C^{\alpha}_{\ell} = \sum_{l=0}^{\ell - 1} q^{\alpha}_l \delta N^{\alpha}_{l+1} $ and $C^{\alpha}_0 = 0$. The dimension is then $1+1+1+q^{\alpha}_{\ell}+( 1+ q^{\alpha}_{\ell}+ 1+ (q^{\alpha}_{\ell})^{2})$. Even for $q^{\alpha}_{\ell}=20$, the corresponding space is already much too big to construct a reasonable grid on it. We therefore suggest using a neural network approximation. 
Let us consider a family of bounded continuous functions $\{\phi_{\x}, \x \in \X\}$, $\X$ being a compact subset of $\Rb^{d_{\X}}$ for some $d_{\X}\ge 1$, such that, for all $q, \delta q\le \Ns$ and $N, \delta N\ge 1$, 
$$
\phi_{\cdot}(\delta q, \delta N,q,N,C,\cdot):(\x,\mu,\ps)\in \X \times \Rb^{q}\times \Rb^{3+q+q^{2}} \mapsto \phi_{\x}(\delta q, \delta N,q,N,C,\mu,\ps)\in \Rb\;\mbox{ is continuous.}
$$
We then fix a family $\{\alpha^{k}\}_{k\le \bar k}$ of deterministic paths of $\Ac(0)$ and simulate independent copies $\{\tilde \theta^{j}\}_{j\le \bar j}$ of $\tilde \theta$ according to $\nu_{0}$, a Normal-inverse-Wishart distribution ${\cal NW}^{-1}(\ps_{0})$. For each $j$, we consider an i.i.d.~sequence $(P_{j'}^{j,1}, \ldots, P_{j'}^{j,\Ns})_{j'\ge 1}$ in the law $\Nc(\tilde \mu^{j},\tilde \Sigma^{j})$ with $\tilde \theta^{j}=:(\tilde \mu^{j},\tilde \Sigma^{j})$. We take these sequences independent and independent of $\tilde \theta$. 
For each $k$ and $j$, we denote by $(\hat \mu^{k,j}_{\ell})_{\ell \le L}$, $(\tilde \ps^{k,j}_{\ell})_{\ell \le L}$ and $(\Ik_{\ell}^{k,j})_{\ell\le L}$ the paths $(\hat \mu^{\alpha^{k}}_{\ell})_{\ell \le L}$, $(\ps^{\alpha^{k}}_{\ell})_{\ell \le L}$ and $(\Ik_{\ell}^{\alpha^{k}})_{\ell \le L}$ associated to the $j$-th sequence $(P_{j'}^{j,1}, \ldots, P_{j'}^{j,\Ns})_{j'\ge 1}$ and the control $\alpha^{k}$. Similarly, we write $f_{p}^{{\rm ad},k,j}(\ell,\cdot)$ to denote the function $f_{p}^{\rm ad}(\ell,\cdot)$ defined as in \eqref{eq: def fap} but in terms of $\Ik_{\ell-1}^{k,j}$ in place of $\Ik_{\ell-1}^{\alpha}$. 
Given an integer $r\ge 1$, we first compute $\check \x_{L - 1}$ as the argmin over $\x \in \X$ of 
\begin{equation*} \label{eq:expectation_last_step}
\sum_{k=1}^{\bar k} \sum_{j=1}^{\bar j} \left| \Eb^{\nu^{k,j}_{L - 1}}_{L - 1}\left[\left| \frac{1}{\Nw} \sum_{i\in \Ik_{L - 1}^{k,j}} (\hat \mu_{L}^{{k},j})^{i} - \tilde \mu^{i} \right|^p\right] - \phi_{\x}(0, {\delta N^{\alpha^{k}}_{L}},q^{\alpha^{k}}_{L - 1},N^{\alpha^{k}}_{L - 1},C^{\alpha^{k}}_{L - 1},\hat \mu^{k,j}_{L - 1},  \ps^{k,j}_{L - 1})  \right|^{r} 
\end{equation*}
in which $\Eb^{\nu^{k,j}_{L - 1}}_{L - 1}$ means that the expectation is taken only over $\tilde \mu$ according to the law $\nu^{k,j}_{L - 1}$, i.e.~${\cal NW}^{-1}(\ps^{k,j}_{ {L-1}})$, and $(\cdot)^{i}$ means that we take the $i$-th component of the vector in the brackets. 
Then, for any $\alpha \in \Ac^{\rm ad}$, 
we set 
 $$
\check \phi_{L-1}(q^{\alpha}_{L-1},N^{\alpha}_{L-1},C^{\alpha}_{L-1}, \cdot)
:=
\min_{{\left(0, \delta N\right)} \in {\rm A}(L-1,\alpha)} \phi_{\check \x_{L - 1}}(0, \delta N, q^{\alpha}_{L-1},N^{\alpha}_{L-1}, {C^{\alpha}_{L-1} }, \cdot), 
 $$
 where 
 $$
 {\rm A}(L-1,\alpha):=\{(\delta q,\delta N)\in \{0\}\times {\mathbb N} :  C^{\alpha}_{L-1}+  \Nw \delta N \le K \}.
 $$

Given $\check \phi_{\ell+1}$ for some $\ell\le L-2$, we then compute a minimizer $\check \x_{\ell}\in \X$ of 
\begin{equation*} \label{eq:expectation_non_last_steps}
\begin{split}
\sum_{k=1}^{\bar k} \sum_{j=1}^{\bar j} &\bigg| \Eb^{\nu^{k,j}_{\ell}}_{\ell}\left[\check \phi_{\ell+1}( q^{\alpha^{k}}_{\ell+1}, N^{\alpha^{k}}_{\ell+1},C^{\alpha^{k}}_{\ell + 1}, \hat \mu^{k}_{\ell+1},\ps^{k}_{\ell+1}) + f^{{\rm ad},k,j}_{p}(\ell+1,\alpha^{k},\tilde \theta) \right] \\
&- \phi_{\x}(\delta q^{\alpha^{k}}_{\ell+1}, \delta N^{\alpha^{k}}_{\ell+1},q^{\alpha^{k}}_{\ell},N^{\alpha^{k}}_{\ell},C^{\alpha^{k}}_{\ell},\hat \mu^{k,j}_{\ell},\ps^{k,j}_{\ell})  \bigg|^{r},
\end{split}
\end{equation*}
 where $\Eb^{\nu^{k,j}_{\ell}}_{\ell}$ means that the expectation is computed over $(\hat \mu^{\alpha^{k}}_{\ell+1},\ps^{\alpha^{k}}_{\ell+1} ,\tilde \theta, \tilde \mk^{\alpha^{k}}_{\ell})$ given $(\hat \mu^{k}_{\ell},\ps^{k}_{\ell} )=(\hat \mu^{k,j}_{\ell},\ps^{k,j}_{\ell})$ and using the prior $\nu^{k,j}_{\ell}$ on $\tilde \theta$ associated to $\ps^{k,j}_{\ell}$. Then, we set 
 $$
\check \phi_{\ell}(q^{\alpha}_{\ell},N^{\alpha}_{\ell},C^{\alpha}_{\ell}, \cdot)
:=
\min_{{\left(\delta q, \delta N\right)} \in {\rm A}(\ell,\alpha)} \phi_{\check \x_{\ell}}(\delta q, \delta N, q^{\alpha}_{\ell},N^{\alpha}_{\ell}, {C^{\alpha}_{\ell}} , \cdot), 
 $$
 where 
 \begin{align*}
 {\rm A}(\ell,\alpha)&:=\{(\delta q,\delta N)\in \[{ {0},q^{\alpha}_{\ell}-  {\Nw} }\times {\mathbb N} :   {C^{\alpha}_{\ell}}+  {(q^{\alpha}_{\ell} -\delta q)}\delta N\le K \},\;\ell<L-2,\\
 {\rm A}(L-2,\alpha)&:=\{(\delta q,\delta N)\in \{q^{\alpha}_{\ell}- \Nw \}\times {\mathbb N} :   {C^{\alpha}_{L-1}}+  {(q^{\alpha}_{L-2} -\delta q)}\delta N\le K \},
 \end{align*}
and so on until obtaining $\phi_{ {0}}( {\Ns,0},0,0,\ps_{0})$. By continuity of $\phi_{\cdot}(\cdot)$ and compactness of $\X$ and ${\rm A}(\ell,\alpha)$ for $\alpha$ given, the minimum is achieved in the above, possibly not unique, and one can choose a measurable map $  {\rm a}^{\star}_{ {\ell}}$ such that 
$$
  {\rm a}^{\star}_{ {\ell}}(q^{\alpha}_{\ell},N^{\alpha}_{\ell},C^{\alpha}_{\ell}, \cdot)\in \mbox{arg}\min_{{\left(\delta q, \delta N \right)} \in {\rm A}(\ell,\alpha)} \phi_{\check x^{N}_{\ell}}(\delta q, \delta N,q^{\alpha}_{\ell},N^{\alpha}_{\ell},  {C^{\alpha}_{\ell}} ,\cdot)
$$
for all $\alpha \in \Ac^{\rm ad}$. Then, given the parameter $\ps_{0}$ of our initial prior $\nu_{0}$, our estimator of the optimal policy is given by $\alpha^{\star}=(q^{\star},N^{\star})$ defined by induction by 
\begin{align*}
(\delta q^{\star}_{1},\delta N^{\star}_{1})&={\rm a}^{\star}_{ {0}}(\Ns,0,0,0,\ps_{0})\;\mbox{ and }\;
(\delta q^{\star}_{\ell+1},\delta N^{\star}_{\ell+1})={\rm a}^{\star}_{ {\ell}}(q^{\star}_{\ell},N^{\star}_{\ell}, C^{\star}_{\ell}, \hat \mu^{\alpha^{\star}}_{\ell},\ps^{\alpha^{\star}}_{\ell})\;\mbox{ for } 0<\ell<L.
\end{align*}

Note that the above algorithm for the estimation of the optimal control only requires off-line simulations according to the initial prior $\nu_{0}$. It is certainly costly but does not require to evaluate the real financial book, it can be trained on a proxy, and can be done off-line. It can be combined with the approach of Remark \ref{rem : tau rho} to reduce the computation time. In order to prepare for the use of a different initial prior, one can also slightly adapt the above algorithm by considering different initial values of $\ps_{0}$ (e.g.~drawn from another distribution around $\ps_{0}$), so as to estimate $\check \phi_{0}$ not only at the point $\ps_{0}$. When applied to the real book, the update of the prior according to \eqref{eq: formule de transition} leads to an additional cost that is negligible with respect to the simulation of the book. It leads to the computation of new priors associated to the financial book at hand, that can be used for a new estimation of the optimal policy or simply as a new initial prior for the next computation of the $\ES$.   

An example of a simple practical implementation is detailed in Appendix \ref{sec: precise implementation}, while numerical tests are performed in Section \ref{sec: tests numeriques}. 
 
\section{Numerical Experiments}\label{sec: tests numeriques}

This section is dedicated to first numerical tests of the different algorithms presented in the previous sections. The settings of the experiments are as follows. We first choose a Normal-inverse-Wishart prior distribution $\nu_0$ with parameters $\ps_{0}:=(\ms_{0},\ks_{0},\is_{0},\Sigma_{0})$. The vector $\ms_0$ is represented on Figure \ref{fig:book_sample_S1_raw_mu} with $\ms_{0}^{i}=\mu^{i}$, $i\le \Ns$, and $\Sigma_0 = (\is_0 - \Ns - 1) \Sigma$ where $\Sigma$ has entries 
\begin{equation} \label{eq:sigma_definition}
\begin{cases}
\Sigma^{ii} = 4.84 \times 10^{12} \textnormal{ if } i = j\\
\Sigma^{ij} = \rho\times 4.84 \times 10^{12} \textnormal{ if } i \neq j, 
\end{cases}
\end{equation}
 {with $\rho=0.6$ or $\rho=0$ depending on the experiments below.}
As for $\ks_0$ and $\is_0$, they are chosen equal to 300, meaning that we have a low confidence in our prior. The computing power is $K = 10^7$. 
\vspace{5mm}

We apply the four different algorithms on 5\,000 runs (i.e.~5\,000 independent implementations of each algorithm). For each run, we 
\begin{itemize} 
\item first simulate a value for the real scenarios and covariance matrices $\left(\tilde{\mu}, \tilde{\Sigma}\right) \sim \mathcal{NW}^{-1}(\ps_0)$,
\item apply each of the four algorithms, with simulated prices following $P_{ {|{\rm s}}} \sim \Nc\left(\tilde \mu, {\tilde \Sigma}\right)$, 
\item for each algorithm, we measure the relative error $\frac{\widehat{ES} - \widetilde \ES}{\widetilde \ES}$ and the error $\ESh - \widetilde\ES$, where $\widetilde \ES = \frac{1}{\Nw}\sum_{i=1}^{\Nw} \tilde{\mu}^{\tilde \mk(i)}$.
\end{itemize}
The four algorithms that we compare are:
\begin{itemize} 
\item A Uniform Pricing Algorithm: All the scenarios are priced with $K / \Ns$ Monte Carlo simulations, and the estimator $\ESh$ is the average of the $\Nw=6$ worst scenarios. This is the most naive method, with only one step and where all scenarios are priced with an equal number of Monte Carlo simulations. It serves as a benchmark.
\item The Heuristic Algorithm: We use the 2-levels strategy described in Section \ref{sec: 2 Level} with the book sample parameters of Table \ref{tab:sample_books_params} and the computation parameters of Table \ref{tab:computing_power}. {We do not evaluate the constant $c$ of Assumption \ref{hyp:sub gamma} but simply set it to $0$, see Remark \ref{rem: c=0}.} The optimal strategy is given by $\left(q_0, q_1, N_1, N_2\right) = (253, 68, 17\,297, 100\,000)$.
\item The Deterministic Algorithm: We run the deterministic algorithm of Section \ref{subsec: DPP determinist} optimized with $\mu=\ms_0$ as the values of the scenarios, $\Sigma$ with $\rho=0.6$ as the covariance matrix and $L=4$. Note that using the real mean parameter as an entry for optimization is quite favorable for this algorithm, although the ``true'' parameter of each run will actually deviate from this mean value. This gives us the strategy $\left(q_0, q_1, q_2, q_3, N_0, N_1, N_2, N_3, N_4 \right) = \left(253, 35, 10, 6, 0, 6\,000, 44\,000, 44\,000, 1\,235\,666\right)$, which we apply to each run.  
\item The Adaptative Algorithm: We do the training part of the adaptative algorithm using our prior $\ps_{0}:=(\ms_{0},\ks_{0},\is_{0},\Sigma_{0})$, with $\rho=0.6$, as parameters and $L=4$. We use a very simple one hidden-layer neural network. It could certainly be improved by using a more sophisticated multi-layers neural network, but this version will be enough for our discussion. Details on the implementation are given in the Appendix \ref{sec: precise implementation}.
 Once this is done, we apply the optimal adaptative strategy on each run. \end{itemize}

  \subsection{Positively correlated scenarios $\rho=0.6$ }
 
 In this first experiment, the simulated runs use the values $\rho=0.6$ and $\is_0 = \ks_0 = 300$. 
 \vspace{5mm}
 
To get an idea of how much noise is added to the average scenario values in our simulations, we plot in Figure \ref{fig:scenario_ranks_and_values_rho_0.6} the prior value $\ms_0^i$ for each scenario of index $i\le \Ns$ (this is the line) and the first $20$ $\tilde \mu^i_j$ out of the $5\,000$ runs for each scenario (these are the points). 

\begin{figure}[H]
  \centering
  \includegraphics[width=\linewidth]{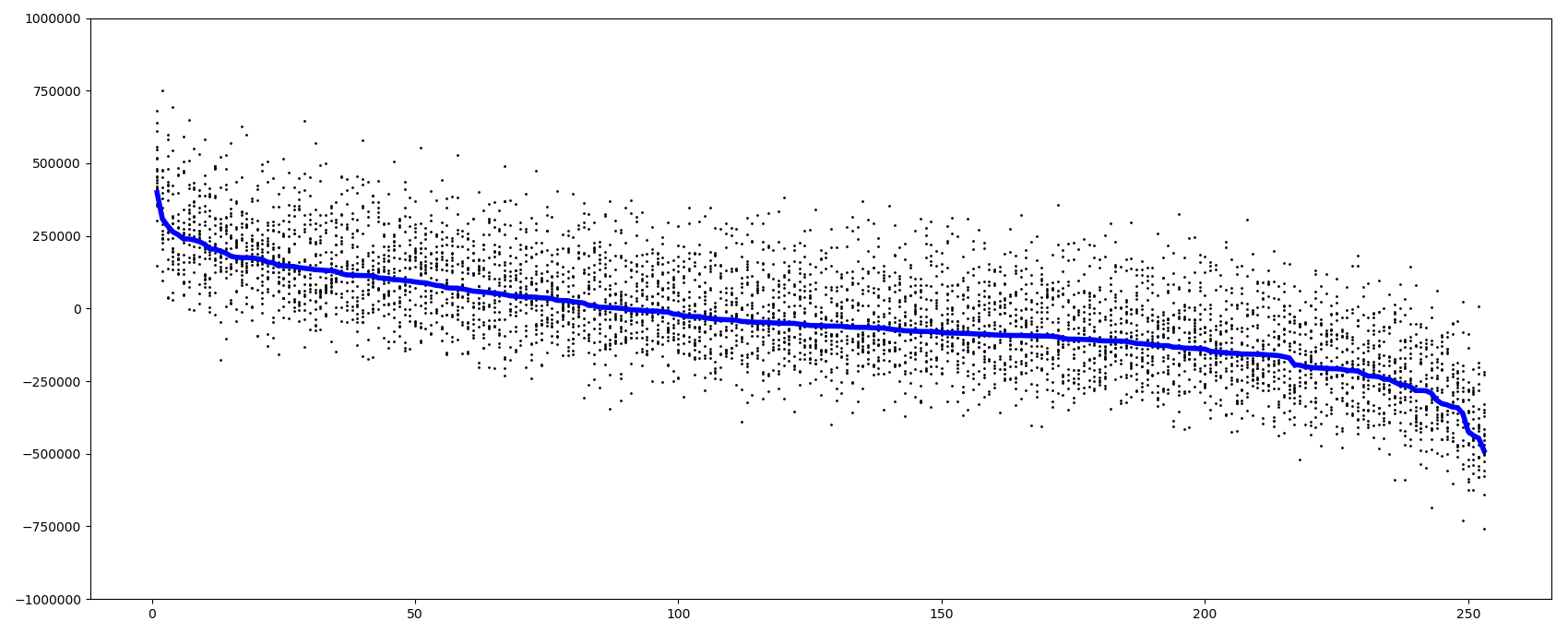}
  \caption{True value of $\mu_{\circ}$ and simulations of $\tilde \mu$}
  \label{fig:scenario_ranks_and_values_rho_0.6}
\end{figure}
 For the adaptative algorithm, the three mostly used strategies are:
\begin{itemize}
\item $\left(q_0, q_1, q_2, q_3, N_1, N_2, N_3, N_4\right) = \left(253, 40, 25, 6, 8\,399, 97\,995, 172\,504, 577\,252 \right)$  
\item $\left(q_0, q_1, q_2, q_3, N_1, N_2, N_3, N_4\right) = \left(253, 40, 30, 6, 8\,399, 99\,733, 148\,560, 608\,040 \right)$  
\item $\left(q_0, q_1, q_2, q_3, N_1, N_2, N_3, N_4\right) = \left(253, 40, 30, 6, 8\,399, 75\,033, 123\,860, 748\,007 \right)$ 
\end{itemize}
Compared to the deterministic algorithm, we see that the adaptative one uses much less Monte Carlo simulations at the final steps and focuses more on the intermediate steps to select the worst scenarios. The deterministic algorithm is also more aggressive in the choice of $q_{1}$ and $q_{2}$. This can be easily explained by the fact that the latter believes that the real distribution is not far from the solid curve on Figure \ref{fig:scenario_ranks_and_values_rho_0.6} (up to standard deviation) while the adaptative one only knows a much more diffuse distribution corresponding to the cloud of points of Figure \ref{fig:scenario_ranks_and_values_rho_0.6} since his level of uncertainty is quite high for our choice $\is_0 = \ks_0 = 300$. 
\vspace{2mm}

 On Figures \ref{fig:relative_error_histogram_ad_rho_0.6}-\ref{fig:relative_error_histogram_uni_rho_0.6}, we plot the histograms of the relative errors. We see that the distribution is tightest for the deterministic algorithm, followed quite closely by the adaptative algorithm. Both of them perform very well. As expected, the uniform algorithm is very poor. Note that the heuristic one already very significantly improves the uniform algorithm, although it does not reach the precision of the two most sophisticated algorithms (without surprise). Because of the huge uncertainty mentioned above, the adaptative algorithm is rather conservative while the deterministic algorithm makes full profit of essentially knowing the correct distribution, and performs better. We will see in our second experiment that things will change when we will deviate from the parameters used for optimizing the deterministic algorithm (by simply passing from $\rho=0.6$ to $\rho = 0$ in the simulated runs).
  
\begin{figure}[H]
\centering
\begin{minipage}[t]{0.35\linewidth}
\includegraphics[width=\linewidth]{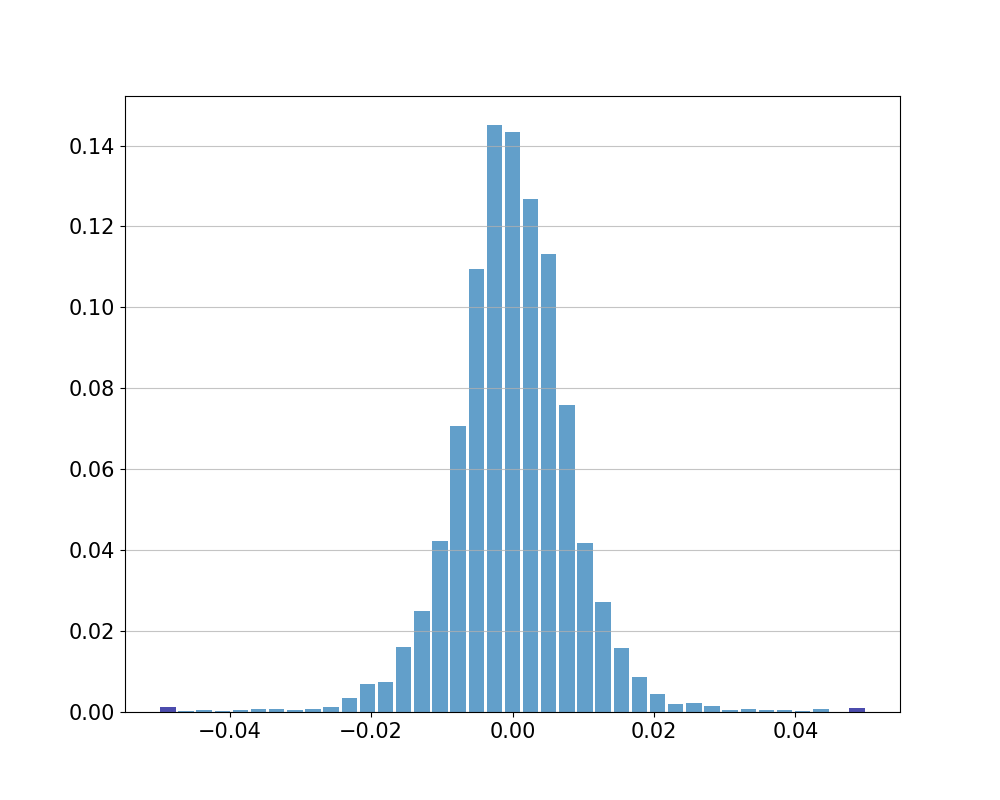}
\caption{Relative Error for Adaptative Algorithm}
\label{fig:relative_error_histogram_ad_rho_0.6}
\end{minipage}
  \begin{minipage}[T]{0.1\linewidth}
  \end{minipage}
\begin{minipage}[t]{0.35\linewidth}
\includegraphics[width=\linewidth]{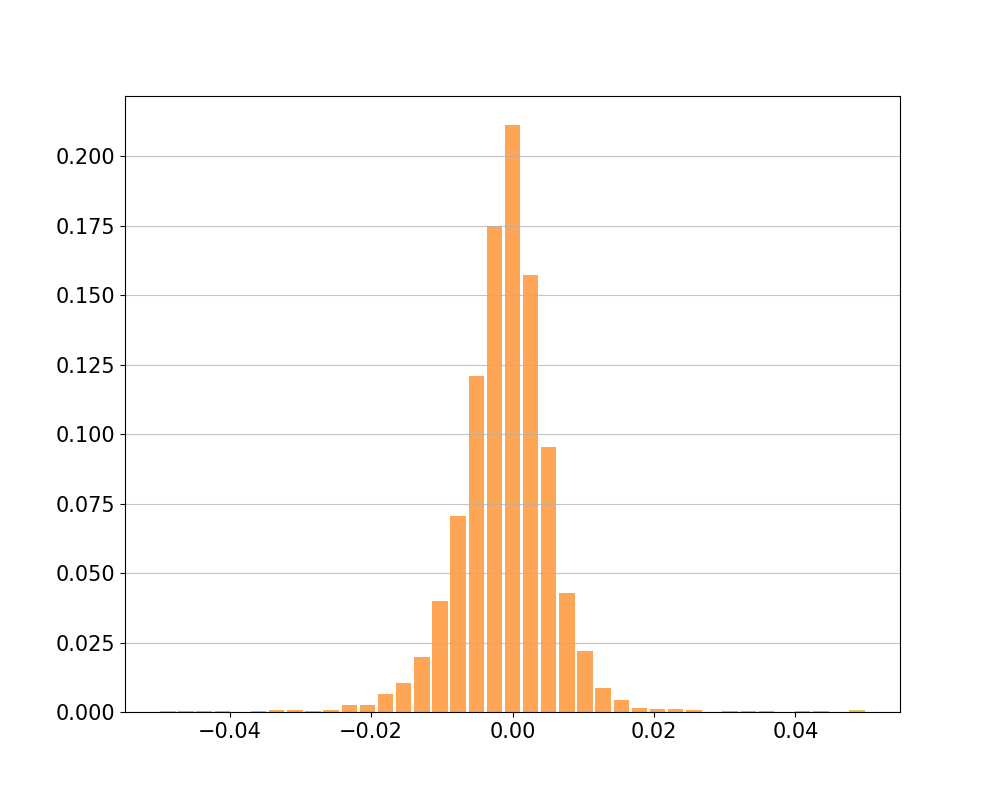}
\caption{Relative Error for Determinist Algorithm}
\label{fig:relative_error_histogram_det_rho_0.6}
\end{minipage}
\begin{minipage}[t]{0.35\linewidth}
\includegraphics[width=\linewidth]{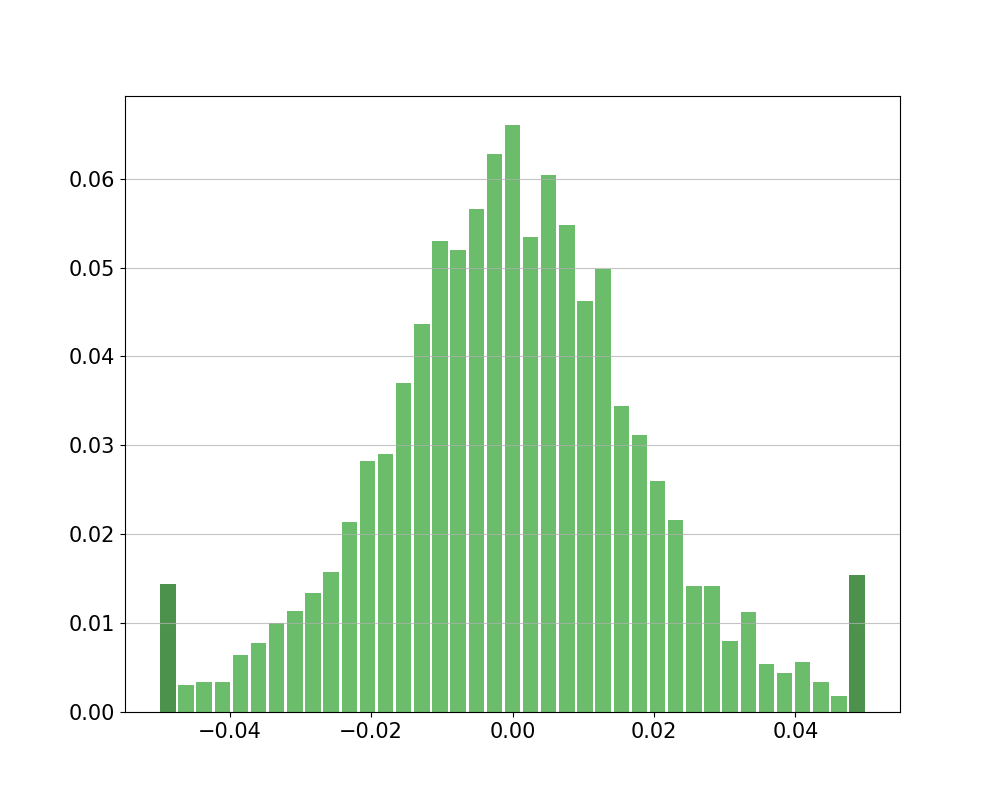}
\caption{Relative Error for Heuristic Algorithm}
\label{fig:relative_error_histogram_heur_1_rho_0.6}
\end{minipage}
\begin{minipage}[T]{0.1\linewidth}
\end{minipage}
\begin{minipage}[t]{0.35\linewidth}
\includegraphics[width=\linewidth]{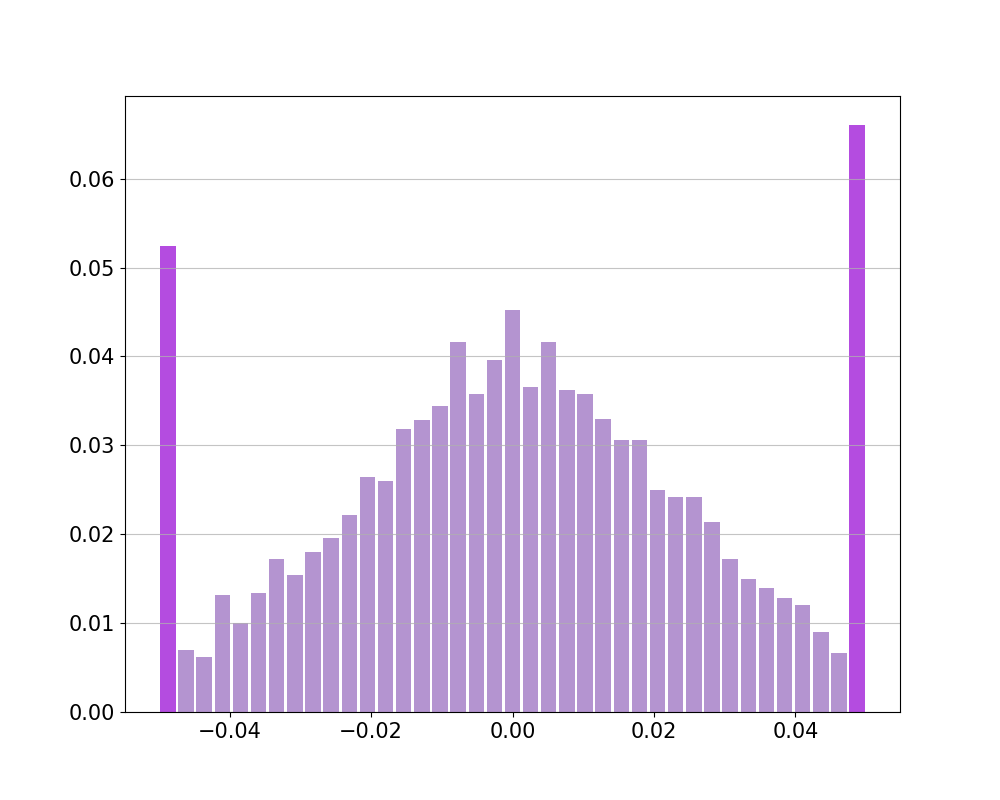}
\caption{Relative Error for Uniform Algorithm}
\label{fig:relative_error_histogram_uni_rho_0.6}
\end{minipage}
\end{figure}

In Table \ref{tab:K_10^7_errors_rho_0.6}, we provide the $\Lb^1$ and relative errors (with standard deviations), the $\Lb^2$ error and the number of correct selections, that is the number of runs for which a given algorithm has chosen the correct worst 6 scenarios. In terms of $\Lb^1$ or $\Lb^2$ error, the relative performance of the algorithms is as above. However, if we look at the number of correct selections, we see that the adaptive algorithm performs better than the other 3 algorithms. Again, by comparing the strategies of the deterministic and the adaptive algorithms, we see that those of the adaptative algorithm are more conservative on the ranking and filtering part versus the final pricing as it puts relatively more Monte Carlo simulations to detect the correct scenarios and relatively less for their estimation. 

\begin{table}[H]
\centering
\begin{tabular}{ |c|c|c|c|c|c|c|}
\hline
Algorithm & $\Lb^1$ Err. & $\Lb^1$ Err. Std & Rel. Err. (\%) & Rel. Err. Std (\%) & $\Lb^2$ Err. & Correct Selections\\
\hline
Ad. Alg. & 1 891 & 20.4 & 0.623 & 0.00886 & 2377 & 4247\\
Det. Alg. & 1 411 & 16.1 & 0.465 & 0.00693 & 1813 & 3499\\
Heur.~Alg. & 4 562 & 50.2 & 1.49 & 0.0234 & 5779 & 4054\\
Unif. Alg. & 7 269 & 81.6 & 2.38 & 0.0348 & 9279 & 3500\\
\hline
\end{tabular}
\caption{Errors for $\rho = 0.6$}
\label{tab:K_10^7_errors_rho_0.6}
\end{table}

In Figures \ref{fig:tail_distribution_error_histogram_ad_rho_0.6}, we plot the function $x \mapsto \mathbb{P}[X > 5000 - x]$ where $X$ is the absolute error of the algorithm on a run. 

\begin{figure}[H]
\centering
\begin{minipage}[t]{\linewidth}
\includegraphics[width=\linewidth]{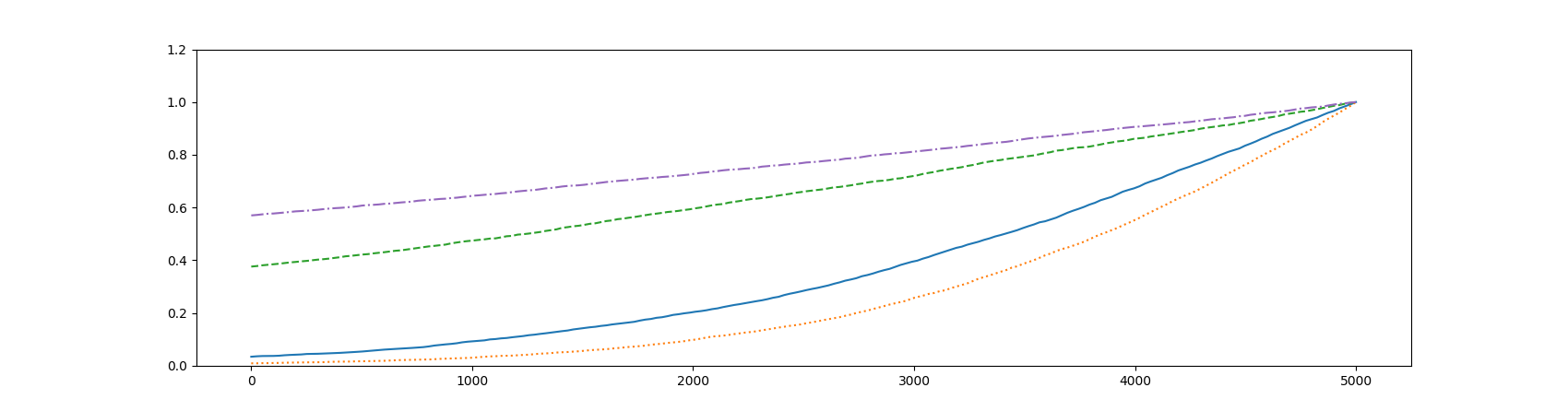}
\caption{Tail Distribution of the errors. First top lines: Uniform and Heuristic algorithms, respectively. Solid line: Adaptative algorithm. Dotted line: Deterministic algorithm. }
\label{fig:tail_distribution_error_histogram_ad_rho_0.6}
\end{minipage}
\end{figure}

In Figure \ref{fig:numbers_graph_rho_0.6}, we provide, for the first 4 runs, the values and real ranks of the 6 worst scenarios selected by each algorithm. The numbers displayed are the true ranks of the selected scenarios given by $\tilde{\mu}$ and their y-coordinate is the value obtained when running the algorithm. ``Real'' is the real values as sampled. 
\begin{figure}[H]
  \centering
  \includegraphics[width=\linewidth]{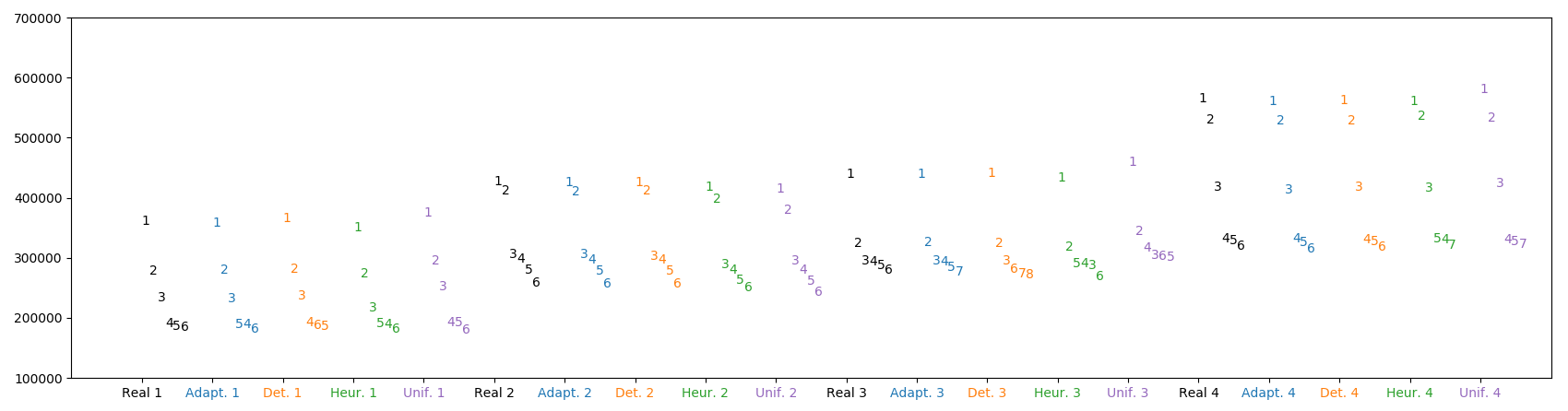}
  \caption{Worst Scenarios Ranks and Values}
  \label{fig:numbers_graph_rho_0.6}
\end{figure}

\subsection{Uncorrelated scenarios $\rho=0$}

We now do the numerical test with $\rho=0$ as the true correlation. The deterministic and adaptative algorithm are still trained with $\rho=0.6$, but $P_{|{\rm s}}$ is simulated using $\rho=0$.
 
On Figures \ref{fig:relative_error_histogram_ad_rho_0}-\ref{fig:relative_error_histogram_uni_rho_0}, we show the histograms of the relative errors. We see that the distribution of the relative errors is now tightest for the adaptative method, followed by the deterministic method, then by the heuristic and the uniform methods. Furthermore, we see that the distribution corresponding to the deterministic method is significantly biased to the left. This is actually true for all algorithms, but at a less significant level. This suggests that we now have a large part of the error that does not come from the final pricing error, but from errors in the selection of scenarios.

\begin{figure}[H]
\centering
\begin{minipage}[t]{0.35\linewidth}
\includegraphics[width=\linewidth]{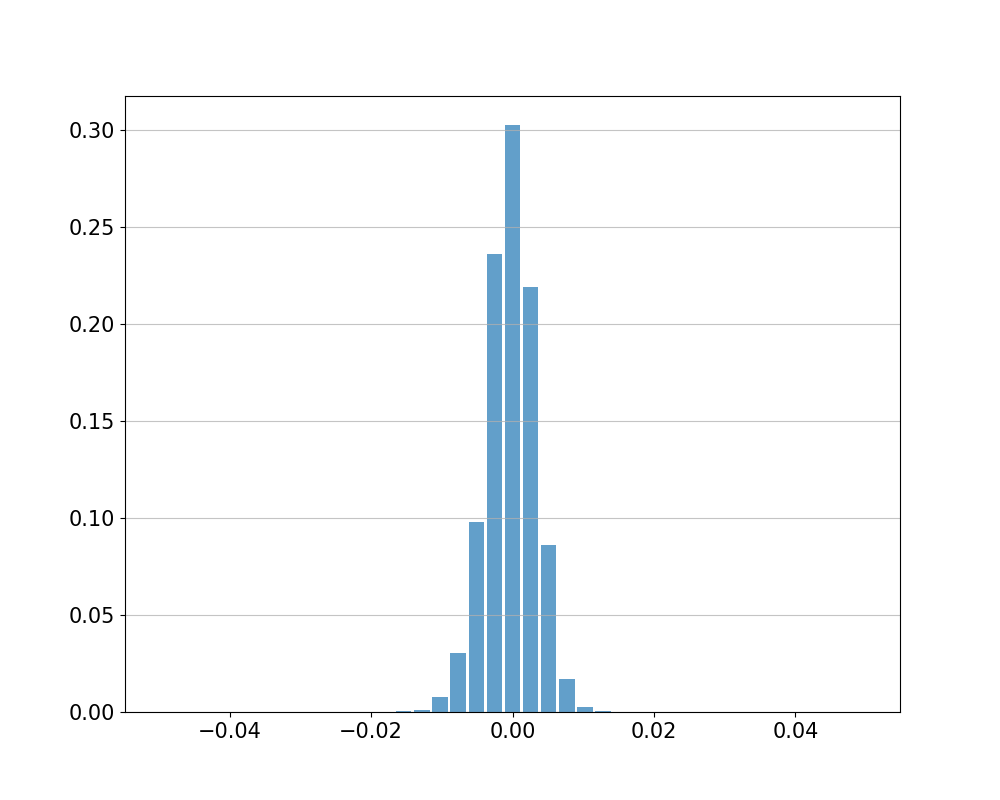}
\caption{Relative Error for Adaptative Algorithm}
\label{fig:relative_error_histogram_ad_rho_0}
\end{minipage}
  \begin{minipage}[T]{0.1\linewidth}
  \end{minipage}
\begin{minipage}[t]{0.35\linewidth}
\includegraphics[width=\linewidth]{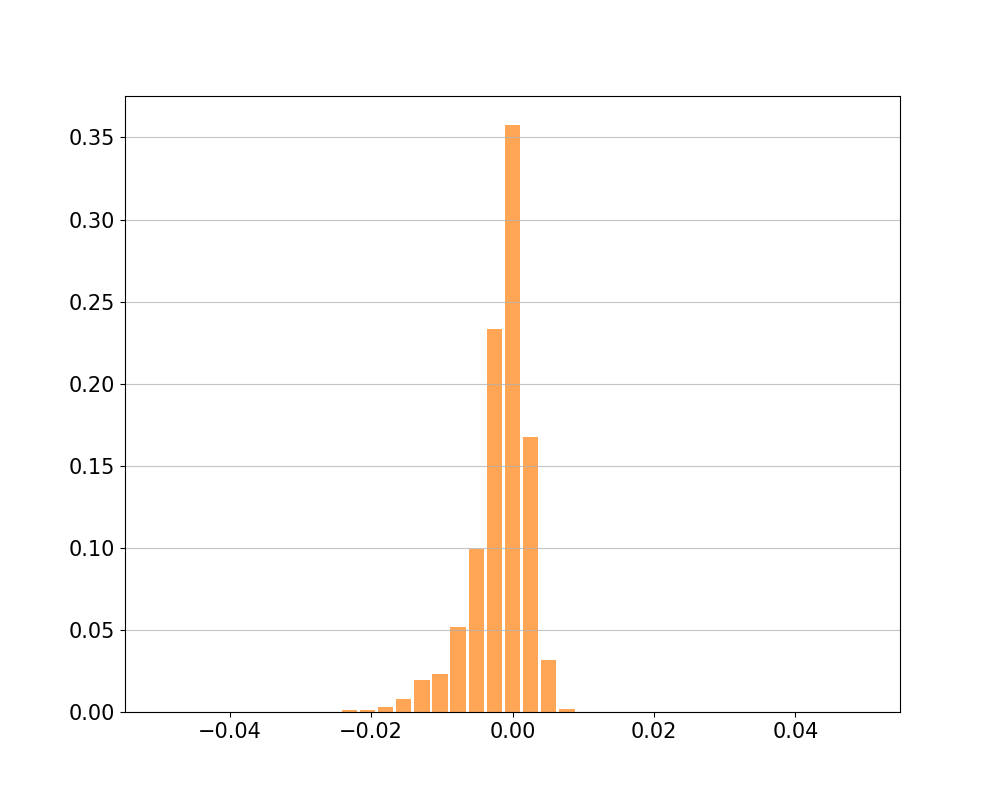}
\caption{Relative Error for Determinist Algorithm}
\label{fig:relative_error_histogram_det_rho_0}
\end{minipage} 
\begin{minipage}[t]{0.35\linewidth}
\includegraphics[width=\linewidth]{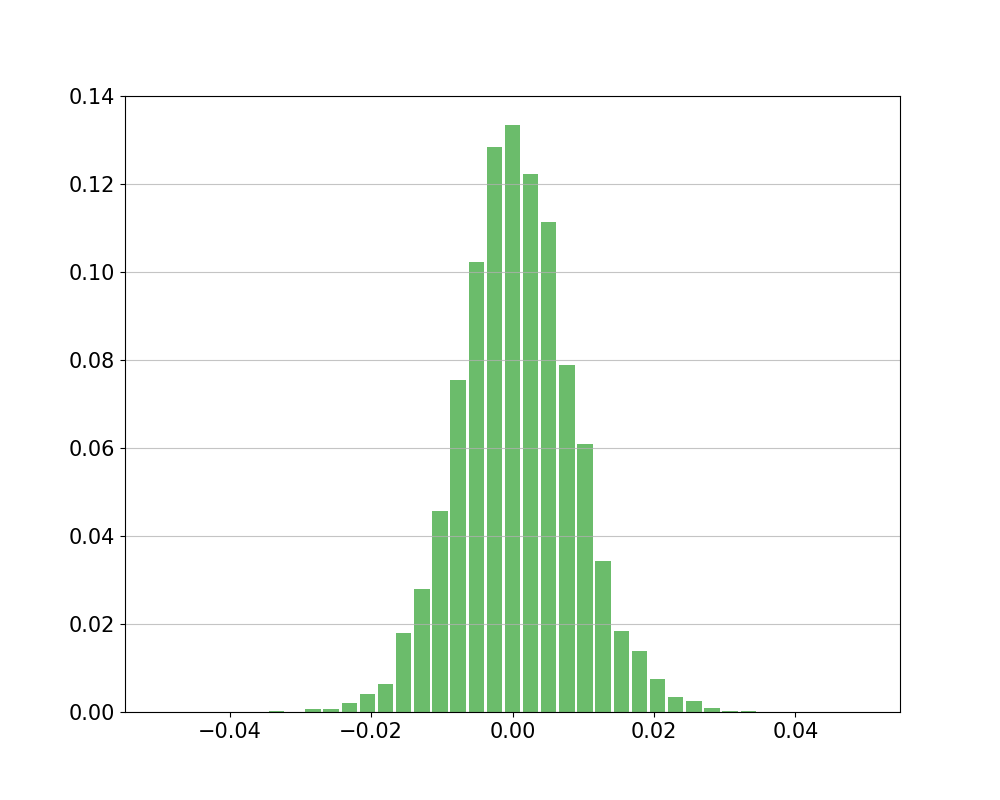}
\caption{Relative Error for Heuristic Algorithm}
\label{fig:relative_error_histogram_heur_1_rho_0}
\end{minipage}
\begin{minipage}[T]{0.1\linewidth}
\end{minipage}
\begin{minipage}[t]{0.35\linewidth}
\includegraphics[width=\linewidth]{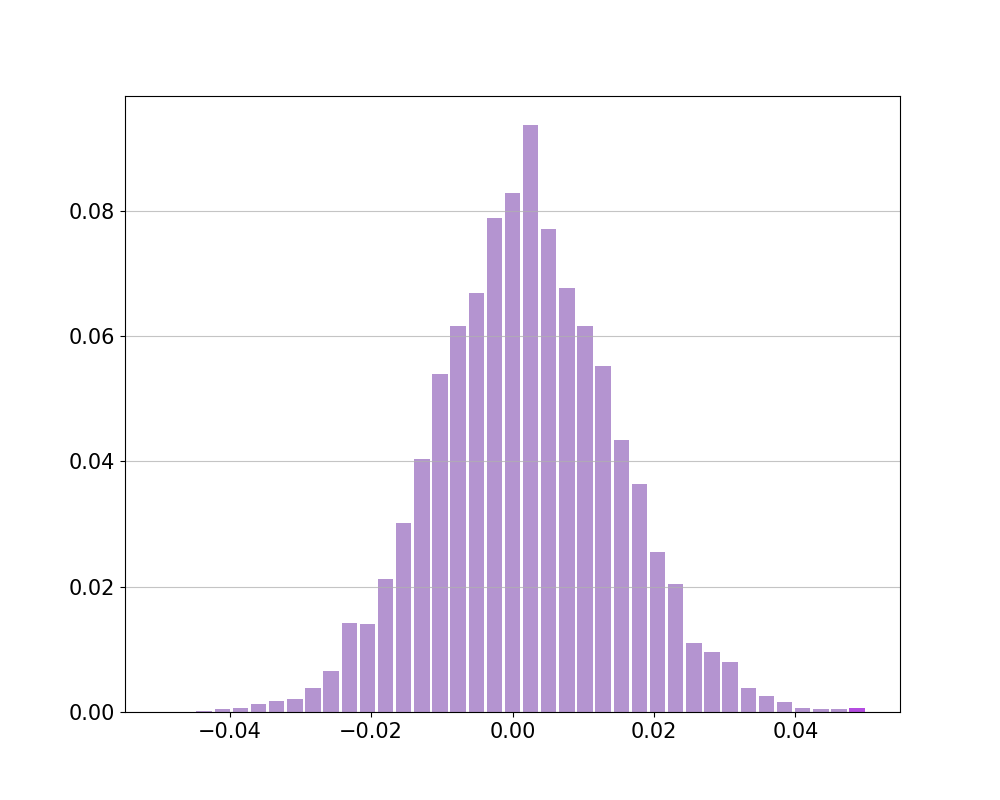}
\caption{Relative Error for Uniform Algorithm}
\label{fig:relative_error_histogram_uni_rho_0}
\end{minipage}
\end{figure}

In Table \ref{tab:K_10^7_errors_rho_0}, we provide the $\Lb^1$ and relative errors (with standard deviations), the $\Lb^2$ error and the number of correct selections for the 4 algorithms. For all algorithms, compared to the case $\rho=0.6$, we see that we have simultaneously a lower number of correct selections of scenarios (which we could expect to increase the errors) and a lower $\Lb^1$ error. This surprising result is explained by the fact that lowering the correlation has two effects. The filtering and ranking part of the algorithm becomes harder, as can be seen from Corollary \ref{cor: borne si sub gamma}. This explains why the number of correct selections becomes lower. However, we compute at the end an average over the $\Nw$ worst scenarios and the error on this average is lower when the pricings are uncorrelated compared to the case where they exhibit a positive correlation. 
 \vspace{5mm}

The adaptative algorithm has now simultaneously the lowest $\Lb^1$ and $\Lb^2$ errors, as well as the highest number of correct selections. We see that it is especially good in $\Lb^2$ error, so we expect it to present a very low number of large errors. As, by construction, it has been trained to detect misspecifications of the parameters, it now has a clear advantage on the deterministic algorithm which does not see it. This results in an improvement of almost 20\% of the $\Lb^2$ error.

 Following the above reasoning, we understand that, compared to the previous experiment, the final pricing error now plays a smaller role and the ranking and selection error a bigger role, which explains why the histogram of the errors for the determinist algorithm is strongly biased to the left, as it now incorrectly selects scenarios more often.

\begin{table}[H]
\centering
\begin{tabular}{ |c|c|c|c|c|c|c|}
\hline
Algorithm & $\Lb^1$ Err. & $\Lb^1$ Err. Std & Rel. Err. (\%) & Rel. Err. Std (\%) & $\Lb^2$ Err. & Correct Selections\\
\hline
Ad. Alg. & 1 083 & 11.8 & 0.27 & 0.00294 & 1 366 & 3 930\\
Det. Alg. & 1 175 & 17.5 & 0.293 & 0.00448 & 1 705 & 3 202\\
Heur. Alg. & 2 547 & 28.33 & 0.628 & 0.00700 & 3 240 & 3 753\\
Unif. Alg. & 4 062 & 44.7 & 1.00 & 0.0111 & 5 147 & 3 102\\
\hline
\end{tabular}
\caption{Errors for $\rho = 0$}
\label{tab:K_10^7_errors_rho_0}
\end{table}

In Figures \ref{fig:tail_distribution_error_histogram_ad_rho_0}, we plot the function $x \mapsto \mathbb{P}[X > 5000 - x]$ where $X$ is the absolute error of the algorithm on a run. As was suggested by the $\Lb^2$ errors of Table \ref{tab:K_10^7_errors_rho_0}, we see that the tail distribution of errors is lowest for the adaptative algorithm, followed by the deterministic algorithm (for big errors), and then by the heuristic and uniform algorithms.

\begin{figure}[H]
\centering
\begin{minipage}[t]{\linewidth}
\includegraphics[width=\linewidth]{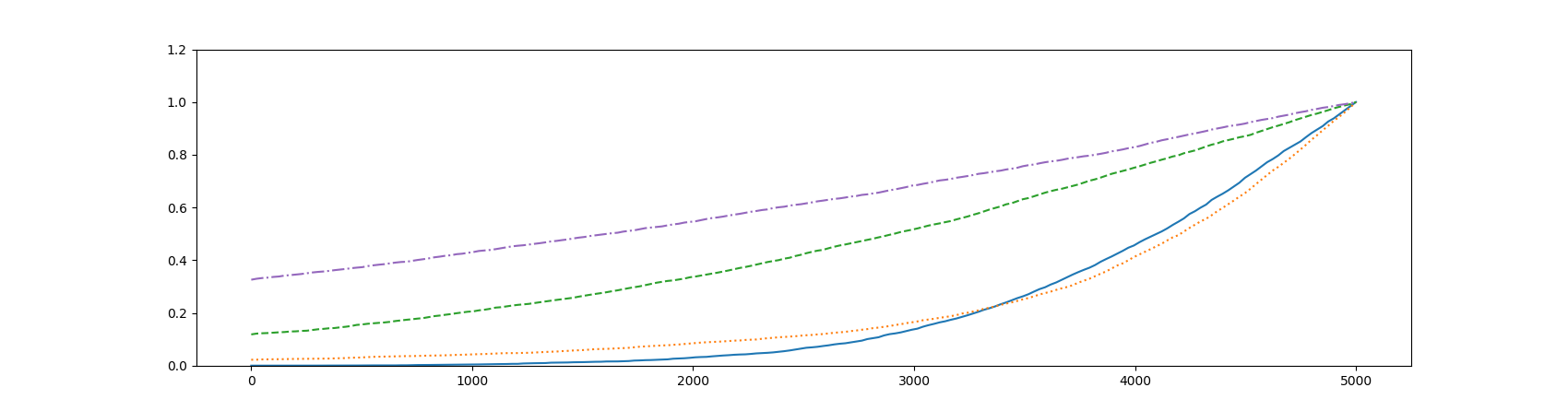}
\caption{Tail Distribution of the errors. First top lines: Uniform and Heuristic algorithms, respectively. Solid line: Adaptative algorithm. Dotted line: Determinist algorithm}
\label{fig:tail_distribution_error_histogram_ad_rho_0}
\end{minipage}
\end{figure}
 
 \section{Conclusion}
 
 We propose in this paper different algorithms for the computation of the expected shortfall based on given historical scenarios. All are multi-steps algorithms that use Monte Carlo simulations to reduce the number of historical scenarios that potentially belong to the set of worst scenarios. We  provide explicit error bounds and we test them on simulated data deviating from the true values of the historical impacts used for computing the associated optimal strategies. The first algorithm is a very easy to implement $2$-steps algorithm that already provides relatively small errors on our numerical tests. A four step deterministic dynamic programming algorithm performs very well when real datas are not far from the parameters used in the optimization procedure. It seems even to be quite robust, as shown by our numerical test in the case where the true correlation parameter is not the one used for computing the optimal policy. Finally, we propose an adaptative algorithm that aims at learning the true value of the parameters at the different steps of the algorithm. Our first numerical tests suggest that it is more conservative than the deterministic one, but probably more robust to parameters misspecifications, as expected. The version we use is built on a very simple one hidden layer neural network and can certainly be considerably improved for industrial purposes.

\newpage
\appendix
 
\section{Proxy of the optimal strategy for the heuristic \eqref{simplified_2_steps_optimization}}\label{sec: h1 explicit}
 
 In the case $p=1$, \eqref{simplified_2_steps_optimization} can even be further simplified by using the upper-bound
\begin{align}\label{eq: borne h1 h2}
\tilde h^{1}_{0}(q_{1})\le \max\{\tilde h_{1}(q_{1});\tilde h_{2}(q_{1})\}
\end{align}
where 
\begin{align*}
\tilde h_1(q_1) &:= \Ns \left(q_1+1-\Nw \right) \delta_0 \exp \left(- \frac{ \left( K - q_1 N_2 \right) (q_1+1-\Nw)\delta_0}{4 \Ns  c} \right)
\\
\tilde h_2(q_1) &:= \Ns \left( \Ns - \Nw \right) \delta_0 \exp \left(- \frac{ \left( K - q_1 N_2 \right) ((q_1+1-\Nw)\delta_0)^2}{4 \Ns  \overline \sigma^{2}}\right).
\end{align*}

The right-hand side of \eqref{eq: borne h1 h2} is now tractable for minimization. 
Given, 
{\small
\begin{equation}
\begin{cases}
\Delta := \left( K - \left( \Nw - 1 \right) N_2 \right)^2 - \frac{32 \Ns N_2 c}{\delta_0} \\
B := \frac{\bar \sigma^2}{c \delta_0} + \Nw - 1 \\
q_1^{2, *} := \max \left(\frac{\Nw - 1}{3} + \frac{2K}{3 N_2}, \Nw \right) \\
q_1^{1, 1, *} := \max \left(\frac{3 \left( \Nw - 1\right)}{4} + \frac{K - \sqrt{\Delta}}{4 N_2}, \Nw \right)\\
q_1^{1, 2, *} := \max \left(\frac{3 \left( \Nw - 1\right)}{4} + \frac{K + \sqrt{\Delta}}{4 N_2}, \Nw \right)
\end{cases}
\end{equation}
 the optimal policy $q^{h}_{1}$ is defined by the following table\footnote{We optimize here over real positive numbers.}:  
\begin{table}[H]
\centering
\begin{tabular}{ |c|c|c|c|c|c|}
\hline
Cond. on $B$ & Cond. $\Delta$ & Cond. $q_1^{2, *}$ & Cond. $q_1^{1, 1, *}$ & Cond. $q_1^{1,2, *}$ & Choice of $q_1^h$ \\
\hline
$  \geq \Ns$  &  &  &  &  & $q_1^h := q_1^{2, *}$ \\
\hline
$  \leq \Nw$  & $  > 0$  &  &  &  & $q_1^h := \underset{q_1 \in \{\Nw, q_1^{1, 2, *}\}}{\textnormal{argmin}}h_0^1(q_1)$ \\
\hline
$  \leq \Nw$ & $  \leq 0$  &  &  &  & $q_1^h := \Nw$ \\
\hline
$\Nw < \cdot < \Ns$ & $  > 0$ & $ \leq B$ & $ \leq B$ & $  \leq B$ & $q_1^h := \underset{q_1 \in \{q_1^{2, *}, B\}}{\textnormal{argmin}} h_0^1(q_1)$ \\
\hline
$\Nw < \cdot < \Ns$  & $  > 0$ & $ \leq B$ & $ \leq B$ & $  \geq B$ & $q_1^h := \underset{q_1 \in \{q_1^{2, *}, q_1^{1, 2, *}\}}{\textnormal{argmin}} h_0^1(q_1)$ \\
\hline
$\Nw < \cdot < \Ns$  & $  > 0$ & $ \leq B$ & $ \geq B$ & $  \geq B$ & $q_1^h := \underset{q_1 \in \{q_1^{2, *}, B, q_1^{1, 2, *}\}}{\textnormal{argmin}} h_0^1(q_1)$ \\
\hline
$\Nw < \cdot < \Ns$ & $  > 0$ & $ \geq B$ & $  \leq B$ & $  \leq B$ & $q_1^h := B$ \\
\hline
$\Nw < \cdot < \Ns$  & $  > 0$ & $ \geq B$ &  & $  \geq B$ & $q_1^h := \underset{q_1 \in \{B, q_1^{1,2, *}\}}{\textnormal{argmin}} h_0^1(q_1)$ \\
\hline
$\Nw < \cdot < \Ns$   & $  \leq 0$ & $  \leq B$ &  &  & $q_1^h := \underset{q_1 \in \{q_1^{2, *}, B\}}{\textnormal{argmin}} h_0^1(q_1)$ \\
\hline
$\Nw <\cdot < \Ns$ & $  \leq 0$ & $  \geq B$ &  &  & $q_1^h := B$ \\
\hline
\end{tabular}
\caption{Optimal $q_1^h$ for $\tilde h^{1}_{0}$.}
\label{tab:definition_of_q_1_h_all_cases}
\end{table}
} 

For simplicity, let us consider the case $c=0$, see Remark \ref{rem: c=0}. 

On Figure \ref{fig:h_0_vs_q_1}, the square is $q_1^{1, 2, *} = 52.41$, the circle is $q_1^{2, *}=68.33$ and the cross is the real optimum $q_1^*=71$ of $h^{1}_{0}$, for the parameters of Tables \ref{tab:sample_books_params} and \ref{tab:computing_power}. We see that we actually almost reach the correct minimum. It corresponds (up to rounding) to $N_1^{1, *} = 23723$, $N_1^{2, *} = 17148$, $N_1^*=15934$.

\begin{figure}[H]
  \centering
  \includegraphics[width=0.6\linewidth]{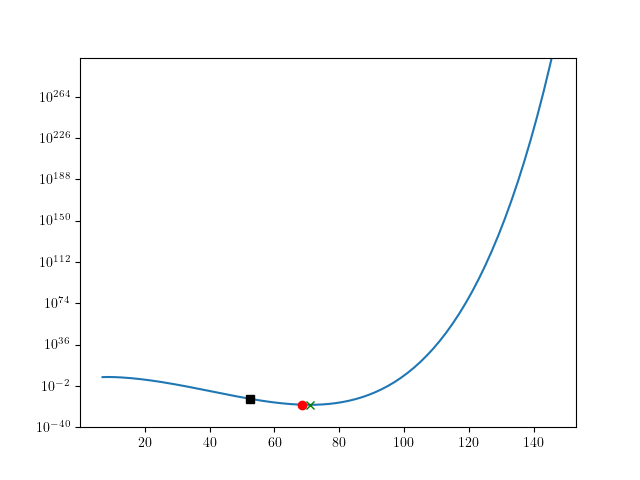}
  \caption{Square: $h_0^1(q_1^{1, 2, *})$. Circle: $h_0^1(q_1^{ 2, *})$. Cross: $h_0^1(q_1^{*})$. }
  \label{fig:h_0_vs_q_1}
\end{figure}

Using the same set of parameters, we plot on Figure \ref{fig:h_0_vs_max_h_1_h_2} the two functions $h_0^1$ and $\tilde h_1$. Although these two functions have values of different orders of magnitude, their shapes are quite close, which explains why we manage to obtain a relatively good approximation for the minimizer.

\begin{figure}[H]
  \centering
  \includegraphics[width=0.6\linewidth]{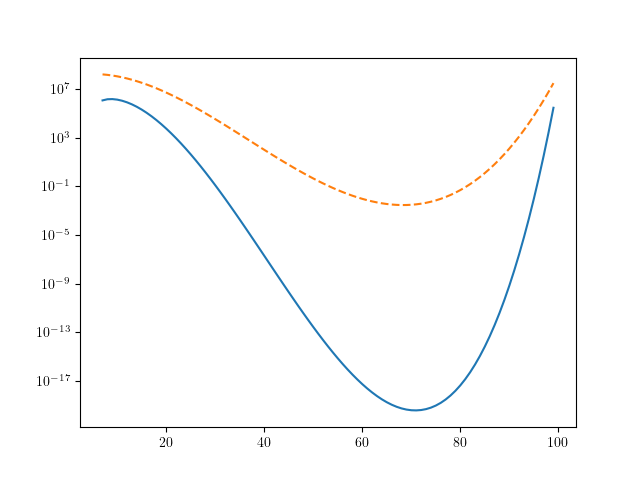}
  \caption{Solid line : $h_0^1$. Dashed line : $\tilde h_1$.}
  \label{fig:h_0_vs_max_h_1_h_2}
\end{figure}

\section{Precise implementation of the neural network algorithm}\label{sec: precise implementation}

In this Appendix, we describe in more details how the neural network approximation of the optimal policy of the adaptative algorithm is constructed. All the parameters values are given in Tables \ref{tab:renorm_constants}, \ref{tab:learning_rates_K_10**7} and \ref{tab:implementation_parameters} below. 

\subsection{Initialization}

\begin{itemize}
\item In practice, the neural network input's size depends on the window size $q$. Therefore, we need to train different neural networks for each window size. In order to get enough points to train each of these neural networks, we have chosen the grid 
\begin{equation*}
q_g=[6, 10, 15, 20, 25, 30, 35, 40, 45, 50, 60, 70, 80, 90, 100, 150, 200, 253]
\end{equation*}
of possible values for $q$. 
\item We simulate independent copies $\{\tilde \theta^{j}\}_{j\le {\rm j\_bar}} = \{ ( \tilde \mu^j, \tilde \Sigma^j) \}_{j \leq {\rm j\_bar}}$ of $\tilde \theta$, where {\rm j\_bar} is given in Table \ref{tab:implementation_parameters}. For each $1 \leq j \leq$ {\rm j\_bar}, $\tilde \Sigma^j$ is an inverse-Wishart of parameters $\is_0, \Sigma_0$, and $\tilde \mu^j$ is a Gaussian random vector of mean $\ms_0$ and covariance matrix $\tilde \Sigma^j / \ks_0$. The parameters $\is_0, \ks_0$ and $\Sigma_0$ are defined in Table \ref{tab:implementation_parameters} and \eqref{eq:sigma_definition}, while $\ms^{i}_{0}=\mu^{i}$, $i\le \Ns$, with the $\mu^{i}$'s of Figure \ref{fig:book_sample_S1_raw_mu}. 
\end{itemize}

\subsection{Strategy Generation}
To generate the deterministic strategies $(\alpha^{k})_{k\le {\rm k\_bar}}$, where {\rm k\_bar} is given in Table \ref{tab:implementation_parameters}, we proceed as follows. 
\begin{itemize}
\item For each $1 \leq k \leq${\rm k\_bar}, we simulate $L+1$ uniform random variables $\left( U_{n} \right)_{n=0}^L$ between $0$ and $1$. We sort them in {increasing} order $\left( U_{s(n)} \right)_{n=0}^L$ and define a cost $K_\ell := K (U_{s(\ell)} - U_{s(\ell-1)})$ when $1 \leq \ell \leq L - 1$, and $K_L = K (U_{s(0)} + 1 - U_{s(L - 1)})$. The idea is that we select $L + 1$ points randomly on a cercle of total length $K$: we choose one of these points, and starting from it, the computational power that we will use at each level $1 \leq \ell \leq L - 1$ is the length of the arc between the previous and the next point. For the last step, we take $K$ times the length between the points $L-1$ and $0$, so as to put, in average, twice more computational power on this last step. 

\item Once we have the computational cost for each step, we can choose the $q_\ell$ for each strategy, so that we can deduce $\delta N_{\ell+1} {:=K_{\ell}/q_{\ell}}$. For $\ell=0$, we choose $\textnormal{q\_index}_0 =  {18}$, where $18$ is the number of terms in the grid $q_g$, which therefore gives $q_0 = q_g[\textnormal{q\_index}_0] = \Ns$. For $\ell = L - 1$, we choose $\textnormal{q\_index}_{L - 1}  = 0$, that is, $q_{L - 1} =  \Nw$. For  $1 \leq \ell \leq L - 2$, we choose $\textnormal{q\_index}_{\ell}$ as a random  integer between $[L - \ell, \textnormal{q\_index}_{\ell-1} - 1]$. The choice of $q_\ell$ is then $q_\ell = q_g[\textnormal{q\_index}_{\ell}]$. We check that the sequence $(N_{\ell})_{1\le \ell\le L}$ is non-decreasing. If this is the case, we keep it, if not, we reject it and do another run. 
\end{itemize}

\subsection{Forward Pass}
The next step is to generate all prices and execute for each $k$ and each $j$ the strategy $k$.

\begin{itemize}
\item  For $1 \leq j \leq \textnormal{j\_bar}$, $1 \leq k \leq \textnormal{k\_bar}$ and $1 \leq \ell \leq L$, we simulate $\delta N_\ell^{k}$ Gaussian variables $\left(P_{  j'}^{j,1},\ldots,P_{ j'}^{j,\Ns}\right)_{ j'=N_{\ell-1}^{k}}^{N_\ell^{k}}$ of mean $\tilde \mu^j$ and covariance matrix $\tilde \Sigma^j$ (independently across $j$ and $k$). 

\item We then update $\hat \mu_\ell^{k, j}, \ms_\ell^{k, j}, \is_\ell^{k, j}, \ks_\ell^{k, j}, \Sigma_\ell^{k, j}$ accordingly, recall \eqref{eq: formule de transition}.

\item Updating $\Sigma_\ell^{k, j}$ from level $\ell - 1$ to level $\ell$ can use a lot of memory. Indeed,  $\sum_{j=N^{\alpha}_{\ell}+1}^{N^{\alpha}_{\ell+1}} ({\cal T}^{\alpha}_{\ell+1}(P_{j})-\delta \hat \mu^{\alpha}_{\ell+1})({\cal T}^{\alpha}_{\ell+1}(P_{j})-\delta \hat \mu^{\alpha}_{\ell+1})^{\top}$ consists in $\delta N_{\ell + 1}^\alpha \times |q_{\ell+1}|^2$ terms, which can quickly exceed memory limits. Therefore, we do the sum with only N\_memory\_new\_pricings\_opt terms at a time, see Table \ref{tab:implementation_parameters} below.
\end{itemize}

\subsection{Computation of f\_precompute, running\_costs and admissible\_sets}

\begin{itemize}
\item In order to speed up the computation time, we now precompute several values that will be used many times afterwards. First, we compute  f\_precompute$(\ell,k,j)$ as $f^{\rm ad}_{1}(\ell,\cdot)$ at the point corresponding to $(k,j)$ except that, in the definition of $f^{\rm ad}_{1}(\ell,\cdot)$, we replace the random permutation $\tilde \mk$ by its estimation from the previous step,    $\tilde \mu$  by its average under the posterior distribution at $\ell+1$, and $\tilde \sigma$ by its estimation at step $\ell+1$.

\item We   compute  $\textnormal{running\_cost}(\ell, k) := C^{\alpha^k}_{\ell}$ of each $k$ at step $\ell$. 

\item We restrict the set of possible actions at step $\ell$,   given that we have followed the strategy $k$ so far, to $\textnormal{admissible\_sets}(\ell, k)$ defined as the collection of $\{(\delta q_{\ell+1}^{ k'}, \delta N_{\ell + 1}^{ k'})$, $k'\le  {\rm k\_bar}\}$, such that 
\begin{equation*}
q_{\ell}^k + \delta q_{\ell + 1}^{k'} \in q_g, \;
N_{\ell + 1}^{k'} > N_\ell^{k}, \; 
\textnormal{running\_cost}(\ell, k) + q_\ell^k \delta N_{\ell + 1}^{k'} \leq \underset{1 \leq k'' \leq \bar k}{\max} \textnormal{running\_cost}(\ell + 1, k'').
\end{equation*}
The last condition avoids inducing a strategy with a running cost that is not present in our data set, when doing the one step optimization. 

\end{itemize}
 
\subsection{Computation of the final expectations}

\item We first pre-compute the quantities 
\begin{equation*}
\Eb^{\nu^{k,j}_{L}}_{L}\left[\left| \frac{1}{\Nw} \sum_{i\in \Ik_{L - 1}^{k,j}} (\hat \mu_{L}^{{k},j})^{i} - \tilde \mu^{i} \right|\right]
\end{equation*}
by Monte Carlo using $N_e$ simulations. As the simulation of an inverse-Wishart random variable is significantly slower than the simulation of a Gaussian random variable, we only simulate 1 inverse-Wishart for $N_p$ Gaussians.  The values of $N_e$ and $N_p$ are given by N\_mu\_tildes\_simulated and N\_wishart\_proportion of Table \ref{tab:implementation_parameters}. The estimation is called $\textnormal{expectation}_L^{k, j}$. 

\subsection{Training of the neural network at level $L$}

\begin{itemize}
\item We use a  neural network with one inner layer with 256 neurons and 1 output layer with 1 neuron to fit $(\textnormal{expectation}_L^{k, j})_{j\le {\rm j\_bar},k\le  {\rm k\_bar}}$. The neurons of the inner layer consist of the composition of the softplus function with an affine transformation of the inputs.

\item We initialize the neural network parameters using a Xavier initialization. We then train the neural network by selecting a random new batch every N\_batch\_change\_proportion. This random new batch is composed of the samples indexed by $1 \leq \mk_a(j) \leq \textnormal{j\_batch}$ and strategies indexed by $1 \leq \mk_b(k) \leq \textnormal{k\_batch}$, where $\mk_a$ and $\mk_b$ are uniform random permutations of $[\![1, {\rm j\_bar}]\!]$ and $[\![1, {\rm k\_bar}]\!]$. For each batch, the algorithm used for the training is the standard gradient descent of Tensorflow. We do N\_Iter training steps in total. The learning rate used is given in Table \ref{tab:learning_rates_K_10**7}. In order to bring the input values of the parameters close to $0$ and $1$, we renormalize them according to the values in Table \ref{tab:renorm_constants}. 
\end{itemize}

\subsection{Computation of the expectations at level $L-1$}

We now estimate 
\begin{equation*}
\Eb^{\nu^{k,j}_{L - 1}}_{L - 1}\left[\check \phi_{L}( q^{\alpha^{k}}_{L}, N^{\alpha^{k}}_{L},C^{\alpha^{k}}_{L}, \hat \mu^{k}_{L},\ps^{k}_{L}) \right]
\end{equation*}
where $\check \phi_{L}$ is the fit of $(\textnormal{expectation}_L^{k, j})_{j\le  {\rm j\_bar},k\le  {\rm k\_bar}}$ from the previous step. The most cpu demanding part    is no more the simulation of the inverse-Wisharts, but the updates of the parameters of the inverse-Wishart. Therefore,  we simulate as many Gaussian random variables as inverse-Wishart random variables, with $N_e$  given by N\_mu\_tildes\_simulated\_non\_final\_level of Table  \ref{tab:implementation_parameters}.

For our computations, we need to update  $\Sigma^{k,j}_{L-1}$ to the corresponding posterior parameter according to \eqref{eq: formule de transition}. This can however lead to an enormous amount of multiplications and additions. Therefore, instead of updating the whole matrix, we only   update  the diagonal terms  according to \eqref{eq: formule de transition} and estimate non diagonal terms by keeping the correlation terms equal to  the ones of $\Sigma^{k,j}_{L-1}$. 
This enables us to approximately gain a factor of $q_{L}^k$ in speed in this critical step. \\

\subsection{Training of the neural network at level $L - 1$}

\begin{itemize}

\item To fit the expectation of the previous step, we use a neural network with the same structure as in level $L$, with the same cost function.

\item The initialization, choice of batches, and training of the neural network are the same as for the level $L$. The number of iteration, learning rate, and renormalization constants are given in Tables  \ref{tab:renorm_constants}, \ref{tab:implementation_parameters} and \ref{tab:learning_rates_K_10**7}.
\item We take $\textnormal{j\_batch} = \min \left( \textnormal{j\_batch\_size}, \textnormal{j\_bar} \right)$ and $\textnormal{k\_batch} = \min \left(\textnormal{k\_batch\_size}, \textnormal{k\_bar} \right)$, where j\_batch\_size and k\_batch\_size are defined in Table  \ref{tab:implementation_parameters}. 

\end{itemize}

\subsection{Computation of the expectations at levels $0\le \ell\le L - 2$}

\begin{itemize}
\item The expectations at step $\ell$ are computed by Monte Carlo after replacing the value function at step $\ell+1$ by its neural network approximation, and $f^{\rm ad}_{1}(\ell,\cdot)$ by f\_precompute$(\ell,\cdot)$.
\item We simulate as many Gaussian random variables as inverse-Wishart random variables, with $N_e$  given by N\_mu\_tildes\_simulated\_non\_final\_level of Table  \ref{tab:implementation_parameters}.
\item We not not fully update  $\Sigma^{k,j}_{\ell}$ to the corresponding posterior parameter but proceed as in level $L-1$. 
\end{itemize}

\subsection{Training of neural networks at levels $ 0 \leq \ell \leq L - 2 $}

\begin{itemize}

\item We now   have to optimize over $q_{\ell}^{k} \in q_g$. Therefore, we must now train up to $|q_g|$ different neural networks (with different inputs' sizes). In practice, we   only train neural networks indexed by $q \in (q_{\ell}^k)_{1 \leq k \leq \textnormal{k\_bar} } \subset q_g$, that is, for all the choices of $q$ that are obtained by at least one strategy at level $\ell$. 

\item We must also choose a $\delta N$ that should be added as an entry of the neural network before optimizing. Furthermore, to help the neural networks converge, we decided to add $\textnormal{f\_precompute}(\ell, j, k)$ as an input. 

\item The loss function and the structure of the neural network is as above, and we still use   Xavier initialization, and  bring the inputs of the neural networks to reasonable values close to 0 and 1 by renormalizing them using the constants of Table \ref{tab:renorm_constants}. 

\item Compared to levels $L$ and $L - 1$, the choice of batches is slightly different. Indeed, to train a neural network associated to $q \in q_g$, we  only  use strategies such that $q_\ell^k = q$. To do so, we first define $S_q = \{k \in [\![1, {\rm k\_bar}]\!]  : q_\ell^k = q \}$. We then define $\textnormal{k\_batch} = \min \left(\textnormal{k\_batch\_size}, |S_q|\right)$ and $\textnormal{j\_batch} = \min \left(\textnormal{j\_batch\_size}, {\rm j\_bar} \right)$. We then proceed nearly identically as for levels $L$ and $L - 1$. We select a new batch every N\_batch\_change\_proportion, composed of indices $1 \leq \mk_a(j) \leq \textnormal{j\_batch}$, $1 \leq \mk_b(k)\leq \textnormal{k\_batch}$, where $\mk_a$  and $\mk_b$ are uniform random permutations of $[\![1, {\rm j\_bar}]\!]$ and $S_q$. For each batch, the algorithm used for the training is again the standard gradient descent of Tensorflow.

\item Compared to levels $L$ and $L - 1$, we found that making the neural networks converge was much harder. In particular, the learning rate had to be really  fine tuned. In order to automatize the process, for each $q$, we proceed as follows. We do not intanciate one, but \\ number\_of\_neural\_networks\_for\_learning\_rate\_test neural networks. For each of these neural networks, we do N\_Iter\_learning\_rate\_test training steps, but use different learning rates for each. For the first neural network, we use base\_learning\_rate as the learning rate, for the second, \\ base\_learning\_rate/10, and for the $k$-th, $\textnormal{base\_learning\_rate}/10^{k-1}$. For each of these neural networks, we store at each iteration step the log error. Once the N\_Iter\_learning\_rate\_test training steps have been done for each of these neural networks, we keep the neural network instance that has the lowest average log error. If it is the $k$-th neural network, we then train it again for N\_Iter training steps, using as learning rate $\textnormal{base\_learning\_rate}/10^{k}$. 
 
\end{itemize}

\subsection{Parallelization}

In practice, we parallelize the forward pass according to the strategy indices $k$. We run thread\_batch\_size processes in parallel, where thread\_batch\_size is defined in Table \ref{tab:implementation_parameters}. 
\vs2

At a given level $\ell$, the computation of $\textnormal{expectation}_\ell^{k, j}$ can be parallelized according to the sample indices $j$. In practice, we run number\_of\_threads\_for\_level\_expectations number of processes in parallel, where number\_of\_threads\_for\_level\_expectations is defined in Table \ref{tab:implementation_parameters}. 
\vs2 

For a given level, the training of each neural network corresponding to a given $q \in q_g$ can be done independently. Therefore, at a given level, we multiprocessed our code in order to train all the neural networks in parallel.

\subsection{Normalization constants, implementation parameters, and learning rates}

\begin{table}[H]
\centering
\begin{tabular}{|c|c|c|c|c|c|c|}
\hline
Level & q & m &    {$\Sigma$} & N & running\_cost & f\_precompute \\
\hline
1 & 6 & $10^{6}$ & $10^{12}$ & $10^4$ & $10^7$ & $10^5$\\
2 & 6 & $10^{6}$ & $10^{12}$ & $10^4$ & $10^7$ & $10^5$\\
3 & 6 & $10^{6}$ & $10^{12}$ & $10^4$ & $10^5$ & $10^6$\\
4 & 6 & $10^{6}$ & $10^{11}$ & $10^4$ & $10^5$ & $10^6$\\
\hline
\end{tabular}
\caption{Inputs' renormalization constants by Level}
\label{tab:renorm_constants}
\end{table} 
 
 \begin{table}[H]
\centering
\begin{tabular}{|c|c|c||c|c|c|}
\hline
Level & q & base\_learning\_rate&Level & q & base\_learning\_rate \\
\hline
1 & 6 & $10^{-9}$ &2 & 6 & $10^{-9}$ \\
1 & 10 & $10^{-9}$&2 & 10 & $10^{-9}$ \\
1 & 15 & $10^{-9}$ &2 & 15 & $10^{-9}$ \\
1 & 20 & $10^{-9}$ &2 & 20 & $10^{-9}$ \\
1 & 25 & $10^{-9}$& 2 & 25 & $10^{-9}$  \\
1 & 30 & $10^{-9}$ &2 & 30 & $10^{-9}$  \\
1 & 35 & $10^{-9}$ &2 & 35 & $10^{-9}$ \\
1 & 40 & $10^{-9}$ &2 & 40 & $10^{-9}$ \\
1 & 45 & $10^{-9}$ &2 & 45 & $10^{-9}$  \\
1 & 50 & $10^{-9}$ &2 & 50 & $10^{-9}$ \\
1 & 60 & $10^{-9}$ & 2 & 60 & $10^{-9}$ \\
1 & 70 & $10^{-9}$ & 2 & 70 & $10^{-9}$\\
1 & 80 & $10^{-9}$ &2 & 80 & $10^{-9}$ \\
1 & 90 & $10^{-9}$ & 2 & 90 & $10^{-9}$\\
1 & 100 & $10^{-9}$ & 2 & 100 & $10^{-9}$ \\
1 & 150 & $10^{-9}$ & 2 & 150 & $10^{-9}$ \\
1 & 200 & $10^{-10}$  & 2 & 200 & $10^{-10}$\\
1 & 253 & $10^{-10}$  & 2 & 253 & $10^{-10}$\\
\hline
3 & 6 & $10^{-7}$ &4 & 6 & $10^{-7}$ \\
\hline
\end{tabular}
\caption{Neural network base learning rates}
\label{tab:learning_rates_K_10**7}
\end{table}

\newpage
\begin{table}[H]
\centering
\begin{tabular}{|c|c|}
\hline
Parameter & Value \\
\hline
j\_batch\_size & 4 \\
k\_batch\_size & 4 \\
N\_batch\_change\_proportion & 1 000 \\
N\_iter\_show\_proportion & 100 \\
smaller\_learning\_rate\_proportion & 10 \\
N\_Iter\_smaller\_learning\_rate & 10 000 \\
L & 4 \\
n\_s & 253 \\
n\_w & 6 \\
k\_bar & 200 \\
j\_bar & 40 \\
i\_0 &   300  \\
k\_0 & 300   \\
 $\Sigma_0$ &  $(300 - 253 - 1)\Sigma$   \\
N\_wishart\_proportion & 1 000 \\
N\_mu\_tildes\_simulated & 1 000 000 \\
thread\_batch\_size & 4 \\
number\_of\_threads\_for\_level\_expectations & 4 \\
thread\_batch\_size\_for\_level\_expectations & 4 \\
p & 1 \\
r & 2 \\
c & 0 \\
N\_Iter & 1 000 000 \\
N\_Iter\_learning\_rate\_test & 100 000 \\
number\_of\_neural\_networks\_for\_learning\_rate\_test & 4 \\
K & 10 000 000 \\
N\_mu\_tildes\_simulated\_non\_final\_level & 1 000 \\
 {N\_memory\_new\_pricings\_opt} &  {100} \\
\hline
\end{tabular}
\caption{Implementation parameters}
\label{tab:implementation_parameters}
\end{table}

\end{document}